\documentclass[12pt]{article}
\usepackage{amsfonts,amssymb,color,amsthm,graphicx,hyperref}

\newcommand{\af}{\alpha}
\newcommand{\C}{\mathbb{C}}
\newcommand{\G}{\mathcal{G}}
\renewcommand{\P}{\mathcal{P}}
\newcommand{\Q}{\mathcal{Q}}
\newcommand{\cR}{\mathcal{R}}
\newcommand{\R}{\mathbb{R}}
\newcommand{\T}{\mathcal{T}}
\newtheorem{asspt}{Assumption}
\newtheorem{definition}{Definition}
\newtheorem{corollary}{Corollary}
\newtheorem{prop}{Proposition}

\newcommand{\BE}{\begin{equation}}  
\newcommand{\EE}{\end{equation}}
\begin{document}


\centerline{\bf\large An Evolver program for weighted Steiner trees}
\ \\

\centerline{Henrique Botelho, Francisco Zampirolli, Val\'erio Ramos Batista\footnote{Email: valerio.batista@ufabc.edu.br, \url{https://sites.google.com/site/vramos1970}, Federal University of ABC, av. dos Estados 5001, St Andr\'e-SP, 09210-580, Brazil}}

\begin{abstract}
We present an algorithm to find near-optimal weighted Steiner minimal trees in the plane. The algorithm is implemented in Evolver programming language, which already contains many built-in energy minimisation routines. Some are invoked in the program, which enable it to consist of only 183 lines of source code. Our algorithm reproduces the physical experiment of a soap film detaching from connected pins towards a stable configuration. In the non-weighted case comparisons with GeoSteiner are drawn for terminals that form a pattern.
\end{abstract}
\ \\
{\bf Keywords}: Weighted Steiner Minimal Trees; Surface Evolver.
\\
\noindent
{\bf MSC[2010]:} 68U05


\section{Introduction}
\label{intro}

It is hard to track back the history involving lab experiments and theoretical knowledge. Regarding soap films the first documented results are due to the physicist J. Plateau in the 19th century~\cite{P}. Among others he worked with wire structures dipped in soapy water and described the resulting films as area-minimising surfaces, which characterises local equilibrium states.

Many of such equilibrium states are found in Nature. For instance, in~\cite{KF} the authors show the importance of these surfaces in crystallography, and a still incomplete geometrical classification of crystal structures was later given in~\cite{H}. See a further contribution to this classification in~\cite{Gui}.

Also there are theories of Steiner trees that arose from soap film experiments. Take two parallel plates connected by pins and dip them into soapy water. The resulting film adjusts to an equilibrium state constituted by strips that together connect all of the pins. When looked from above and perpendicularly to the plates, the film represents a graph whose edges and vertices are the strips and their meetings, respectively. 

Some of the meetings occur at the pins, and the film turns into a tree up to flicking some of the strips. In the nomenclature of~\cite{GP} we then get a {\it relatively minimal tree}, and also a {\it Steiner tree} providing the pins have negligible thickness. It is known that such a tree is an equilibrium state and therefore a local minimum of total length. See~\cite{DKR} for nice discussion and pictures, three of which are reproduced in Fig.~\ref{soap}.

\begin{figure}[ht!]
\center
\includegraphics[scale=.39]{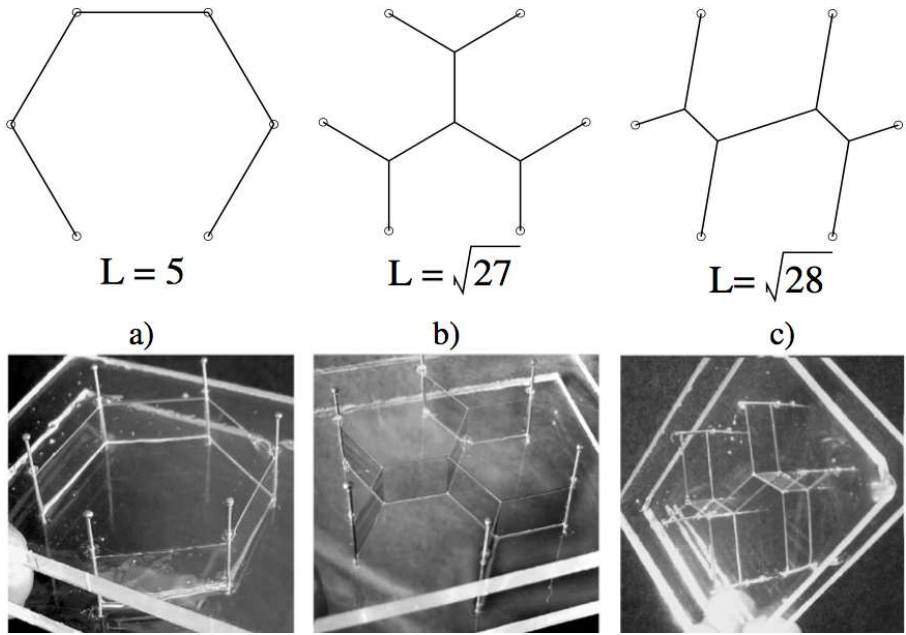}
\caption{Soap Steiner trees of length L for types (a) hexagon; (b) cog; (c) zigzag~\cite[Fig.2]{DKR}.}
\label{soap}
\end{figure}

Of course, it is not always feasible to resort to lab experiments. For arbitrary number and position of pins one would have to spend hours in each case just to observe the resulting Steiner trees. Physical experiments in virtual environment are therefore much easier to perform, and this is one of the objectives of our present paper. Here we consider ideal pins of zero thickness but treat other particularities. For instance, when the plates are dipped the configuration of the film will depend on the manner in which they are removed from the soapy water, and later we shall explain how to customise this procedure with our program.

In our setting the pins belong to the vertices of a tree $\T=(V,E)$, and each edge works as a film strip. However, for any given $V$ our program also seeks the {\it weighted Steiner minimal} $\T$ as defined in~\cite[Sect.2]{CNR}. There the authors prove that any such $\T$ is a Steiner tree, which only coincides with the Steiner minimal tree in very special cases (for instance, when all weights are equal).

Now we are going to use the abbreviations summarised in Table~\ref{abbrev} for the sake of concision:

\begin{table}[ht!]
\centering
\caption{List of abbreviations}\smallskip
\begin{tabular}{|l|c|}\hline
\hspace{2cm}{\bf Name}         & {\bf Initials} \\ \hline\hline
Steiner Tree                   & ST             \\ \hline
Minimal Spanning Tree          & MST            \\ \hline
Steiner Minimal Tree           & SMT            \\ \hline
Weighted Minimal Spanning Tree & WMST           \\ \hline
Weighted Steiner Minimal Tree  & WSMT           \\ \hline
\end{tabular}
\label{abbrev}
\end{table}

Our program runs with the {\it Surface Evolver}~\cite{B}, which is a general-purpose simulator. With Evolver physical experiments can be performed in a completely virtual environment, and one can easily add complexity to the model, or even adapt it for further developments. Firstly introduced in 1989, now Evolver's most recent version is 2.70~\cite{B} with several applications in many Areas of Knowledge like Aerodynamics~\cite{Frank}, Fluid Dynamics~\cite{Chen,Cunsolo,Caio} and Medicine~\cite{Zanka}. Furthermore, Evolver is endowed with several built-in energy minimisation routines which enable saving a lot of programming by just invoking them. Hence we were able to implement our program with only 183 lines of source code. This makes both adaptation and maintenance much easier, as we shall explain in the last section.

This work is organised as follows: Sect.~\ref{prelim} gives some basic notations, definitions and results used throughout the paper. Sect.~\ref{backg} shows quite simple examples that are however essential to understand our strategies. Of course, any heuristic has limitations and ours are discussed in Sect.~\ref{limh}. Then Sect.~\ref{meth} is devoted to explaining our method, which in fact consists not only of the Evolver script but also includes a short graphical input and a preprocessing written in MATLAB/Octave. Finally we show our results and draw conclusions in Sects.~\ref{res} and~\ref{conc}, respectively.

\section{Preliminaries}
\label{prelim}

\begin{definition}
{\rm Consider a graph $G=(V,E)$ in $\R^2$ with a weight function $w:V\to\R_+^*$ and 0-1 adjacency matrix $a_{jk}$. The {\it weighted total length} of $G$ is
\BE
  ||G||=\sum_{j<k}a_{jk}||V_j-V_k||,
\label{met}
\EE
where $||V_j-V_k||=\frac{1}{2}(w_j+w_k)|V_j-V_k|$ is the {\it connection cost} between $V_j$ and $V_k$, and $|V_j-V_k|$ is their Euclidean distance. For the non-weighted case $w\equiv 1$ we write $|G|$ instead of $||G||$ in~(\ref{met}).}
\label{def1}
\end{definition}

As showed in~\cite{CNR}, when $G=(V,E)$ is embedded in $\R^2$ we can reduce $||G||$ as follows: for each adjacent pair $AB,\,BC\in E$ that forms a triangle we take its Euclidean Steiner point $S$, $V'=\{S\}\cup V$, $E'=\{AS,BS,CS\}\cup(E\setminus\{AB,BC\})$, $G'=(V',E')$ and $w(S)=\min w(\{A,B,C\})$. Hence $||G'||\le||G||$ with strict inequality when $S\not\in\{A,\,B,\,C\}$, and this process can be repeated at most $\sharp V-2$ times according to \cite[\S3.4]{GP}. This setting was presented in~\cite{CNR} together with a practical application of the following:

\begin{definition}
{\rm Let $G=\T$ be a tree in Definition~\ref{def1} with $||\T||$ as a global minimum. Namely, there is no other graph connecting its vertices that can reduce $||\T||$. In this case we say that $\T$ is a WSMT.}
\label{def2}
\end{definition}

Notice that we could have $S\in\{A,\,B,\,C\}$ as shown in Fig.~\ref{sliding}, in which case $S$ is called {\it inherent} Steiner point.

\begin{figure}[ht!]
\center
\includegraphics[scale=.33]{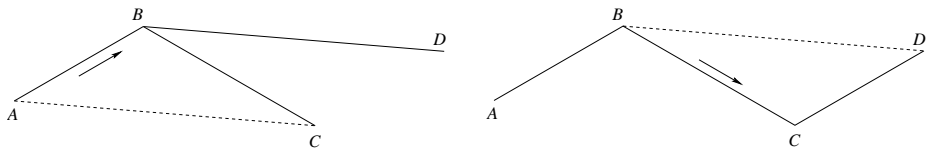}
\caption{Detachment from an initial configuration by sliding to inherent $B$ and $C$.}
\label{sliding}
\end{figure}

As explained in the Introduction our simulation reproduces the physical experiment of getting a soap film after dipping two parallel plates connected by pins in soapy water. Immediately after the plates are taken out the film detaches from an initial configuration in which the pins behave as the vertices of a plane graph $G=(V,E)$, and each edge works as a film strip connecting two pins. The detachment from this initial configuration will result in another plane graph $G'=(V',E')$ where $V'\supset V$ and the elements of $V'\setminus V$ are called (non-inherent) Steiner points. Notice that one might flick some strips for $G'$ to become a tree.

Detachment is due to a physical phenomenon called {\it Marangoni effect}. The forces at a vertex are each parallel to its corresponding incident edge. The intensity of such a force does not depend on the edge length, and therefore we can study the local behaviour at a pin by truncating all of its incident strips to the same length (see Fig.~\ref{mar}(a)). We refer the reader to \cite{Isen} for more details about the Marangoni effect.

Hence $G'$ corresponds to a local minimum of {\it surface tension}, which is directly proportional to $|G'|$. The following proposition gives the geometrical equivalence of this physical fact (see Fig.~\ref{mar}(b) for an illustration):

\begin{figure}[ht!]
\center
\includegraphics[scale=.41]{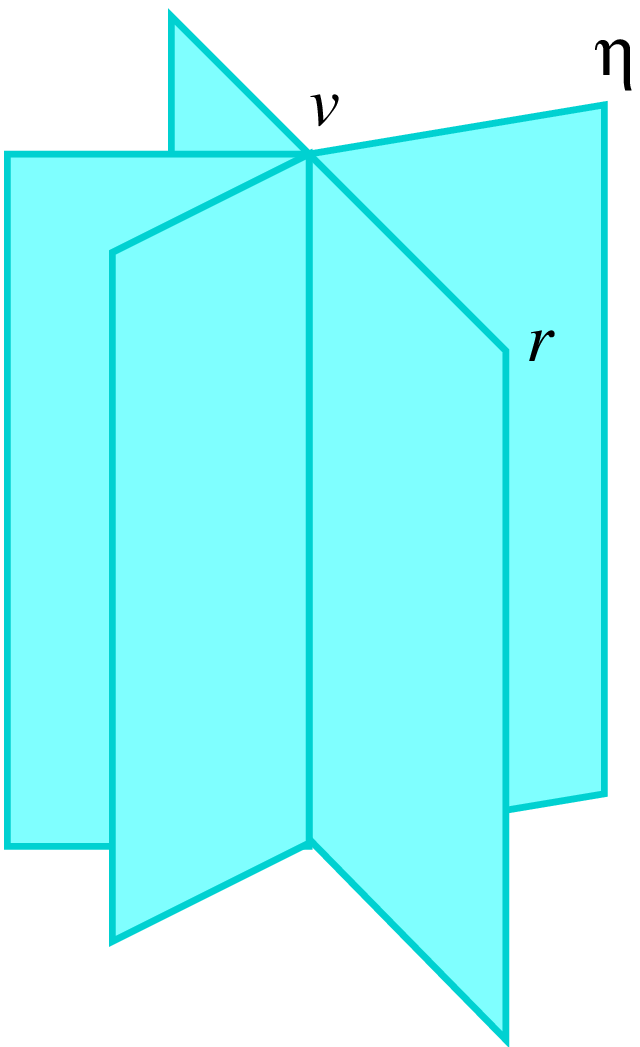}\hspace{1cm}
\includegraphics[scale=.41]{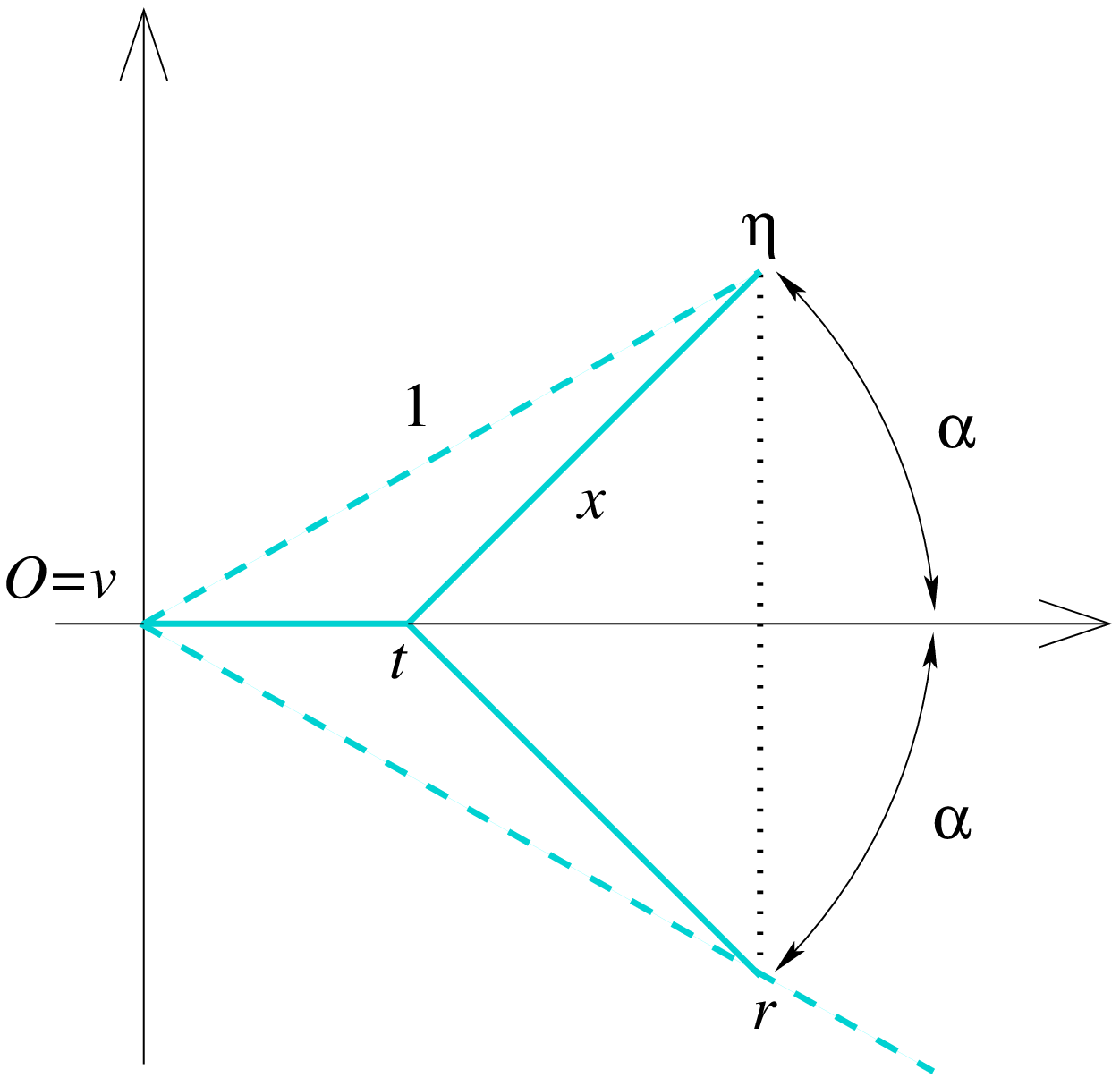}\\

(a)\hfil\hfil(b)\\

\caption{Initial configuration (a) 3d view; (b) view from above.}
\label{mar}
\end{figure}

\begin{prop} Let $G=(V,E)$ be a plane tree $\T$ and take $v\in V$ such that deg$(v)>1$. Hence $|\T|$ can be reduced if and only if $v$ has two adjacent edges that make an angle less than $120^\circ$. 
\label{prop}
\end{prop}

\begin{proof}
Place a Euclidean coordinate system with origin $O=v$ and $Ox$ the bisector of the two adjacent edges. Up to re-scaling the shortest edge measures 1 and it makes an angle $\alpha$ with $Ox$. Fig.~\ref{mar}(b) shows that the initial length 2 turns out to be $f(t)=t+2x$, where $0<t<\cos\alpha$. By the law of cosines we compute $f'$ as
\[
  f(t)=t+2(1+t^2-2t\cos\alpha)^\frac{1}{2}\Longrightarrow
  f'(t)=1+\frac{2(t-\cos\alpha)}{(1+t^2-2t\cos\alpha)^\frac{1}{2}}.
\]

Hence $f'(0)=1-2\cos\alpha<0\iff\alpha<60^\circ$. Notice that $f'(0)$ is positive for $\af>60^\circ$. Since $f''(0)>0$ then $f$ is also increasing for $\af=60^\circ$.
\end{proof}

If $|G|>|G'|$ then this local minimum $G'$ is the nearest to $G$. We get the following result:

\begin{corollary}
The soap film starts detaching from the most acute angles of each $v\in V$, all less than $120^\circ$.
\label{cor}
\end{corollary}

\begin{proof}
In the demonstration of Proposition~\ref{prop} we saw that $f'(0)=1-2\cos\af$, hence the smaller $\af$ the faster $|G|$ decreases to $|G'|$.
\end{proof}

As we have mentioned at the Introduction the detachment of a soap film $G$ from pins is a physical phenomenon that can be performed in a completely virtual environment, and we chose the Surface Evolver for this purpose. Of course, our program considers not only Corollary~\ref{cor} but also further results that we shall present in the next sections.

Because of Proposition~\ref{prop} up to flicking the odd small film strip of $G'$ one gets an ST that we call $\T$. Even if $\T$ is not an SMT a suitable weight function in Definition~\ref{def2} will turn it into a WSMT. Conversely, any WSMT in the metric~(\ref{met}) is an ST (see~\cite{CNR} for details).

Given a set of terminals $V=\{V_1,\cdots,V_n\}$ and $w:V\to\R_+^*$, in order to find a WSMT that minimises~(\ref{met}) we resort to the following strategy: first compute the WMST of $V$ by an adaptation of Prim's algorithm and get an initial tree $\P$, then use Evolver to add and detach vertices from $\P$ according to Proposition~\ref{prop}. As we shall see in the next section this heuristic leads to a $\T$ that will not always be a WSMT. However, if the Gilbert-Pollak conjecture were also valid for weighted trees then a WSMT could be at most $1-\sqrt{3}/2\cong 13.4\%$ shorter than the WMST, for any set $V$, and then our heuristic would still give a tree close to the true WSMT in total length.

But we shall see in Sects.~\ref{meth} and~\ref{res} that the Steiner ratio $\sqrt{3}/2$ does not apply to our case. Anyway, whenever the actual WMST and the plane WMST are still close in length the proposed heuristic can be taken as a satisfactory approach to the actual WSMT. Otherwise we content ourselves with its purpose of reproducing a lab experiment.

\section{Background}
\label{backg}

Let $\C$ denote the standard complex plane $\R\times i\R$ with real and imaginary axes as vertical and horizontal, respectively. The points $\pm1$ and $i\sqrt{3}$ are the vertices of an equilateral triangle whose group of symmetries $\G$ is generated by reflection in $i\R$ and 120$^\circ$-rotation around $i/\sqrt{3}$. For $s\in(0,1/2)$ consider the group orbit $\G(1-2s)$. This orbit has six points indicated with bullets in any item of Fig.~\ref{comprs}.

Let us start with $s<2-\sqrt{3}$, for instance $s=1/8$. Now enumerate the elements of $\G(1-2s)$ as $V_k$, $1\le k\le6$, so that Fig.~\ref{comprs}(a) shows their corresponding MST with $V_3=1-2s$ counterclockwise. We call it $\P$ but for this set of terminals it is Fig.~\ref{comprs}(b) that depicts their SMT, which we call $\Q$.

Their total lengths are $4-2s$ and $2\sqrt{3}$, respectively. Hence $\P$ and $\Q$ are both SMTs for $s=2-\sqrt{3}$, illustrated in Figs.~\ref{comprs}(c,d). Now $\P$ remains the SMT until $\P\cup V_1V_6$ becomes a regular hexagon at $s=1/3$. For $s>1/3$ neither is an SMT but Fig.~\ref{comprs}(f) shows it in dotted line, and now its total length is $\sqrt{3}$. For $s>\sqrt{3}$ we can take the MST $(\P\setminus V_1V_2)\cup V_1V_6$, whose total length is $6-8s$ but just before $\Q$ collapses at $s=1/2$. 

\begin{figure}[ht!]
\center
\includegraphics[scale=.29]{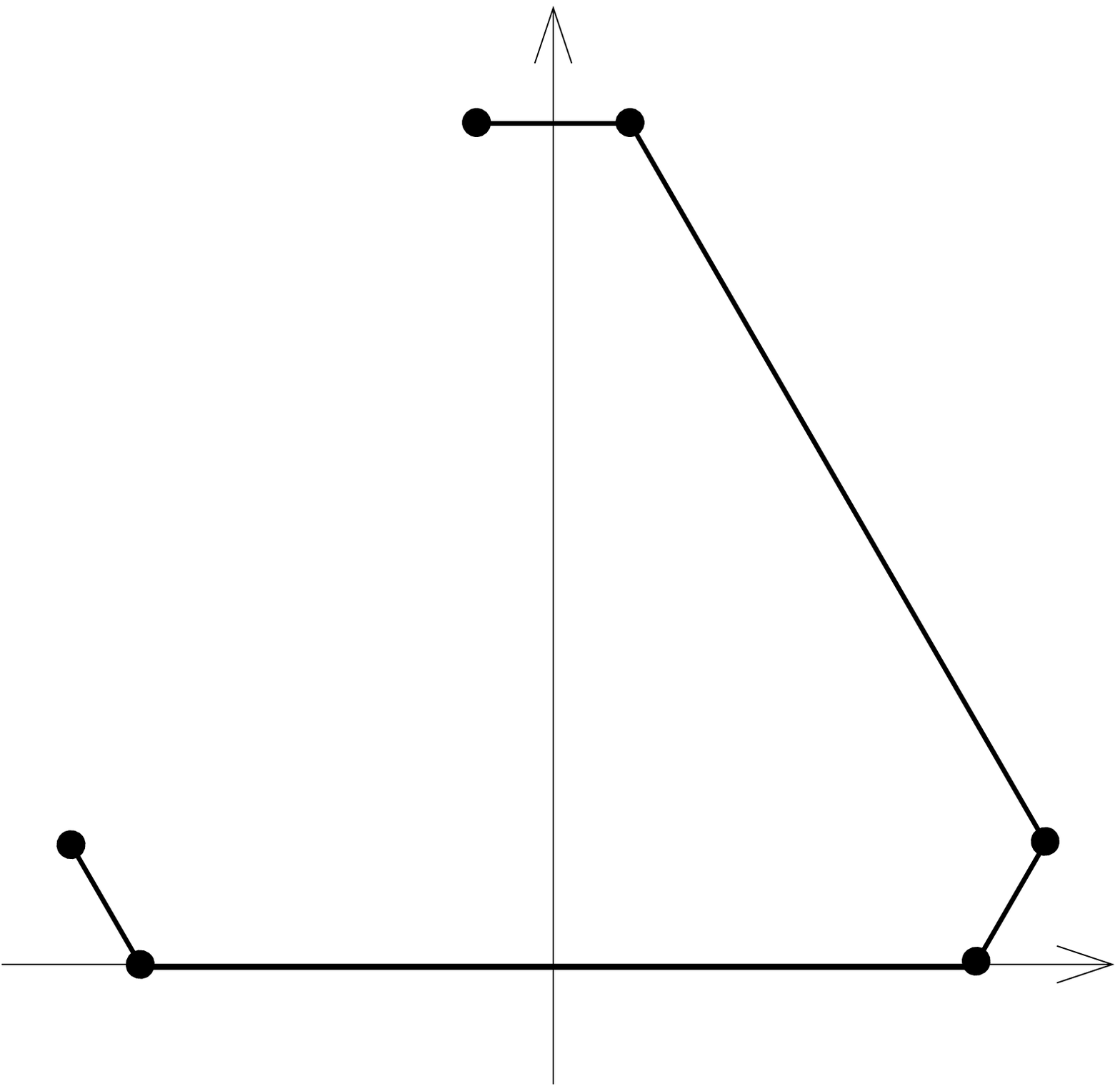}\hspace{1cm}
\includegraphics[scale=.29]{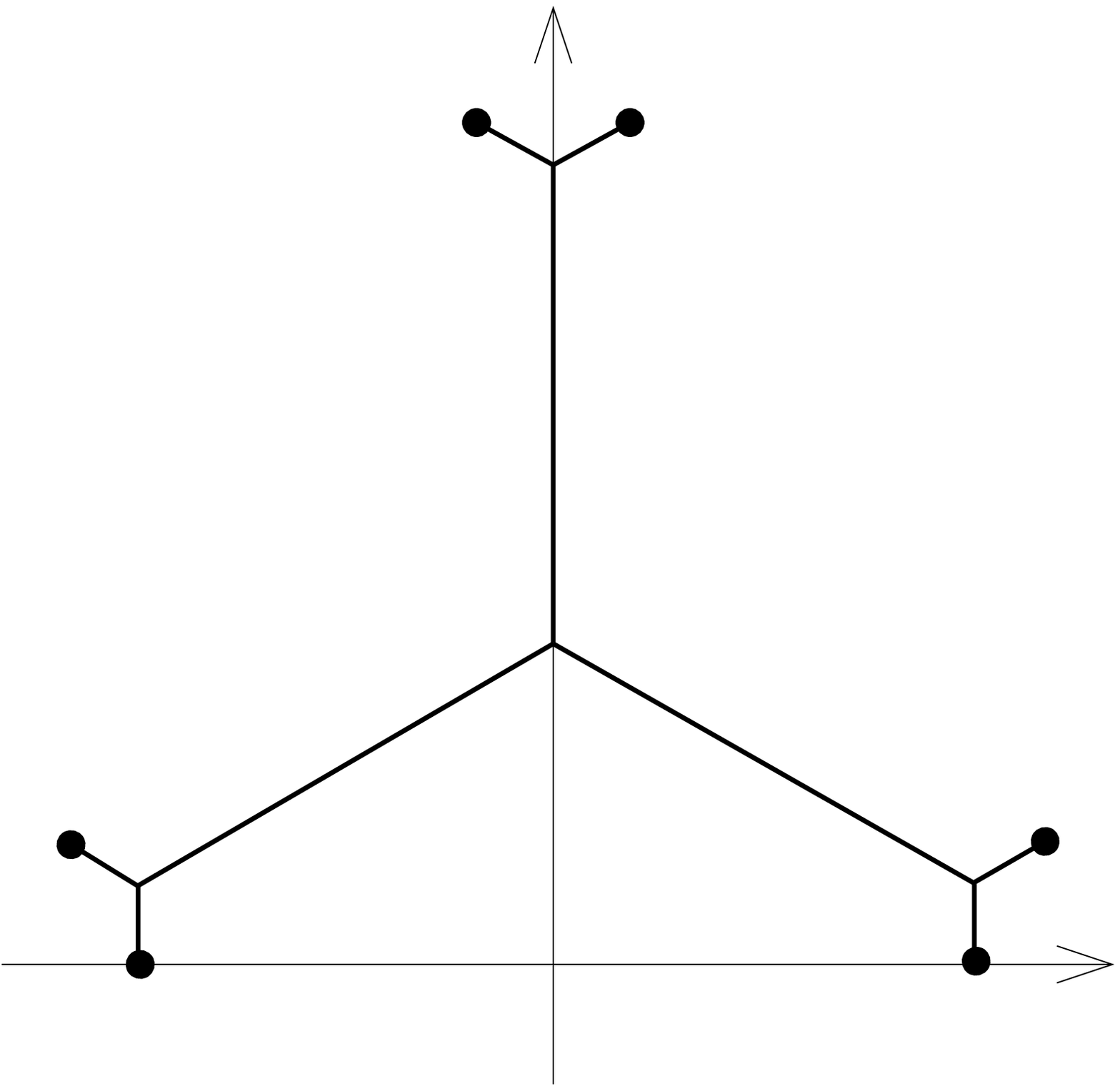}\\

(a)\hfil\hfil(b)\\

\includegraphics[scale=.29]{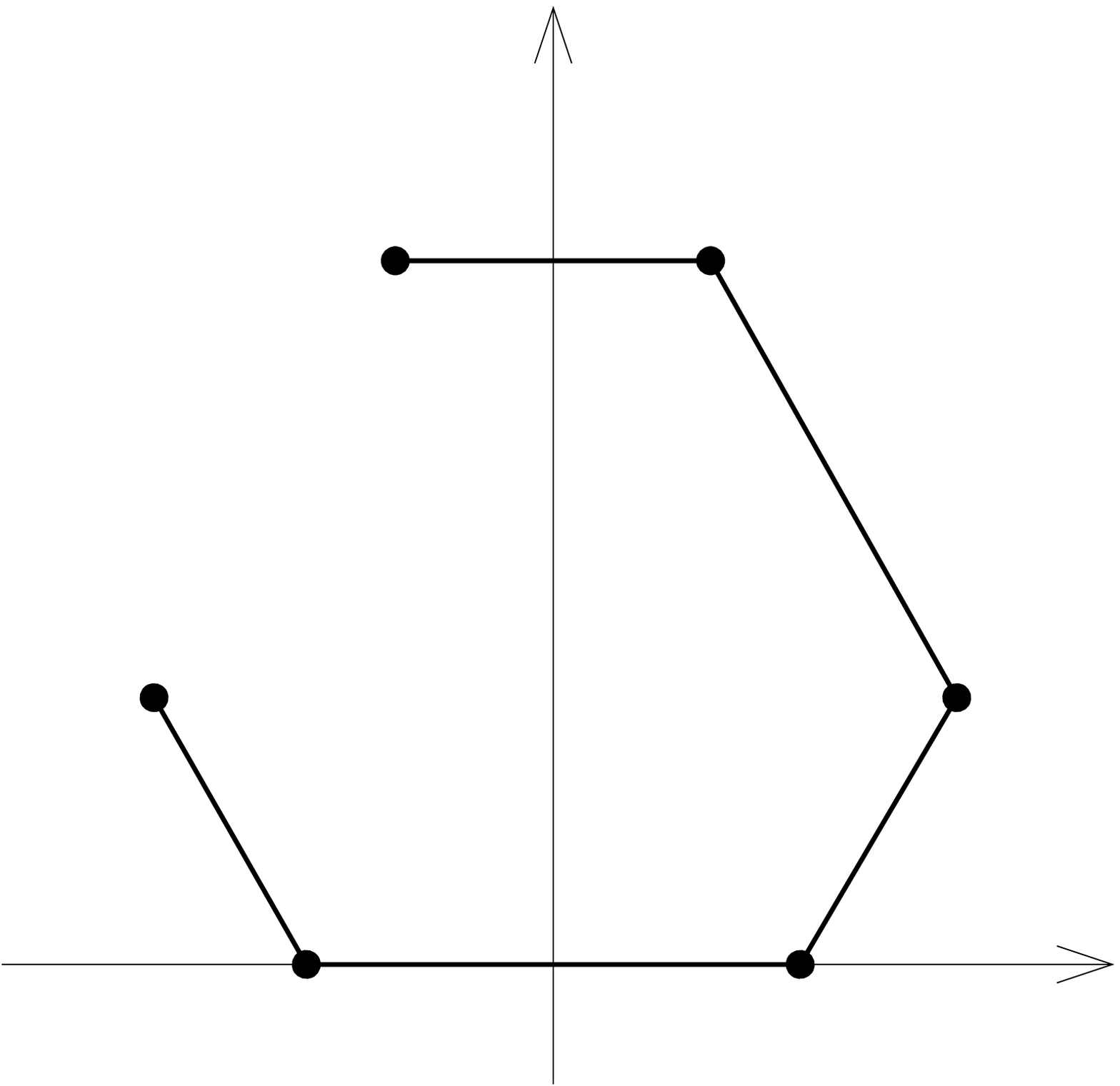}\hspace{1cm}
\includegraphics[scale=.29]{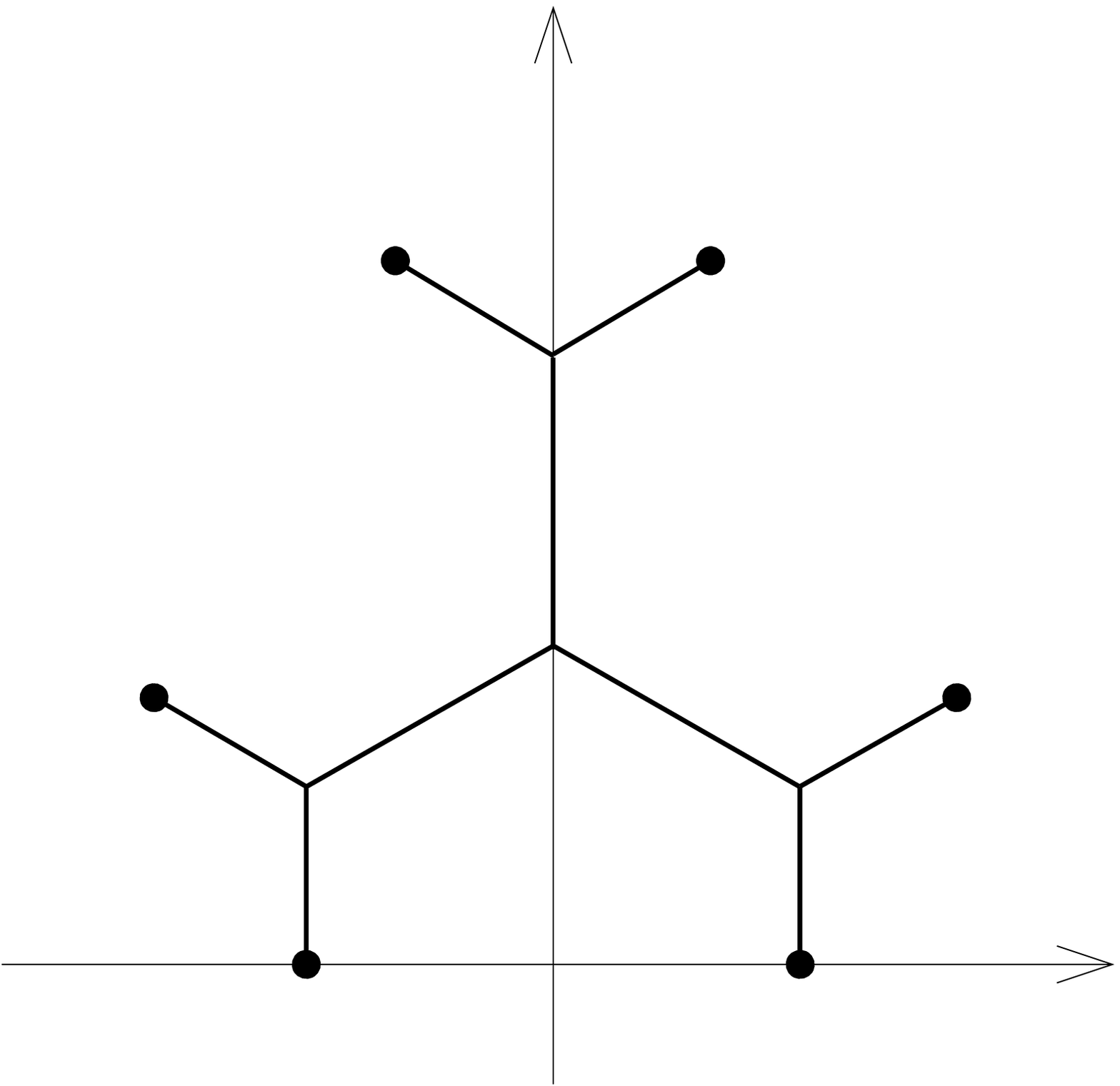}\\

(c)\hfil\hfil(d)\\

\includegraphics[scale=0.29]{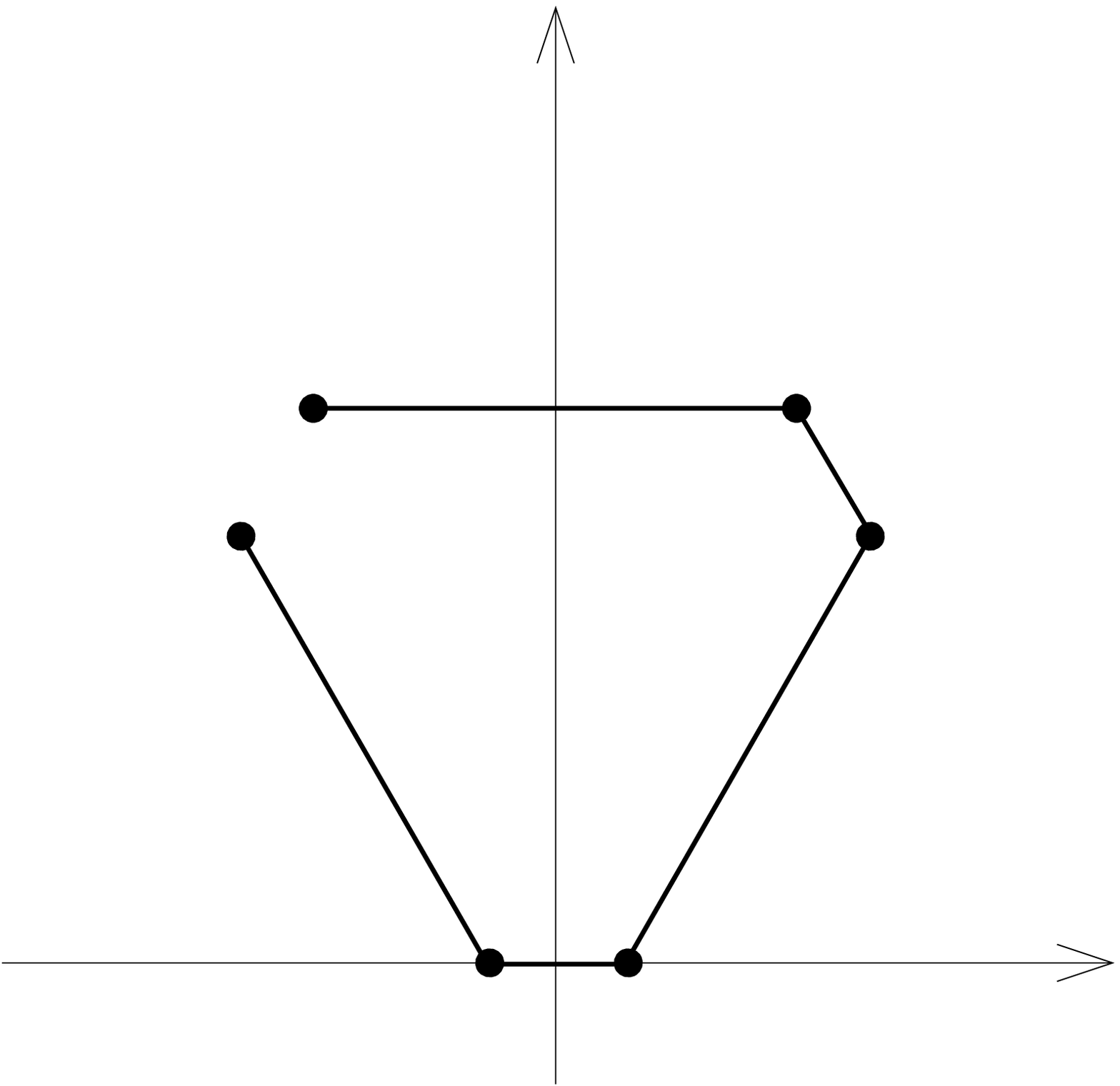}\hspace{1cm}
\includegraphics[scale=0.29]{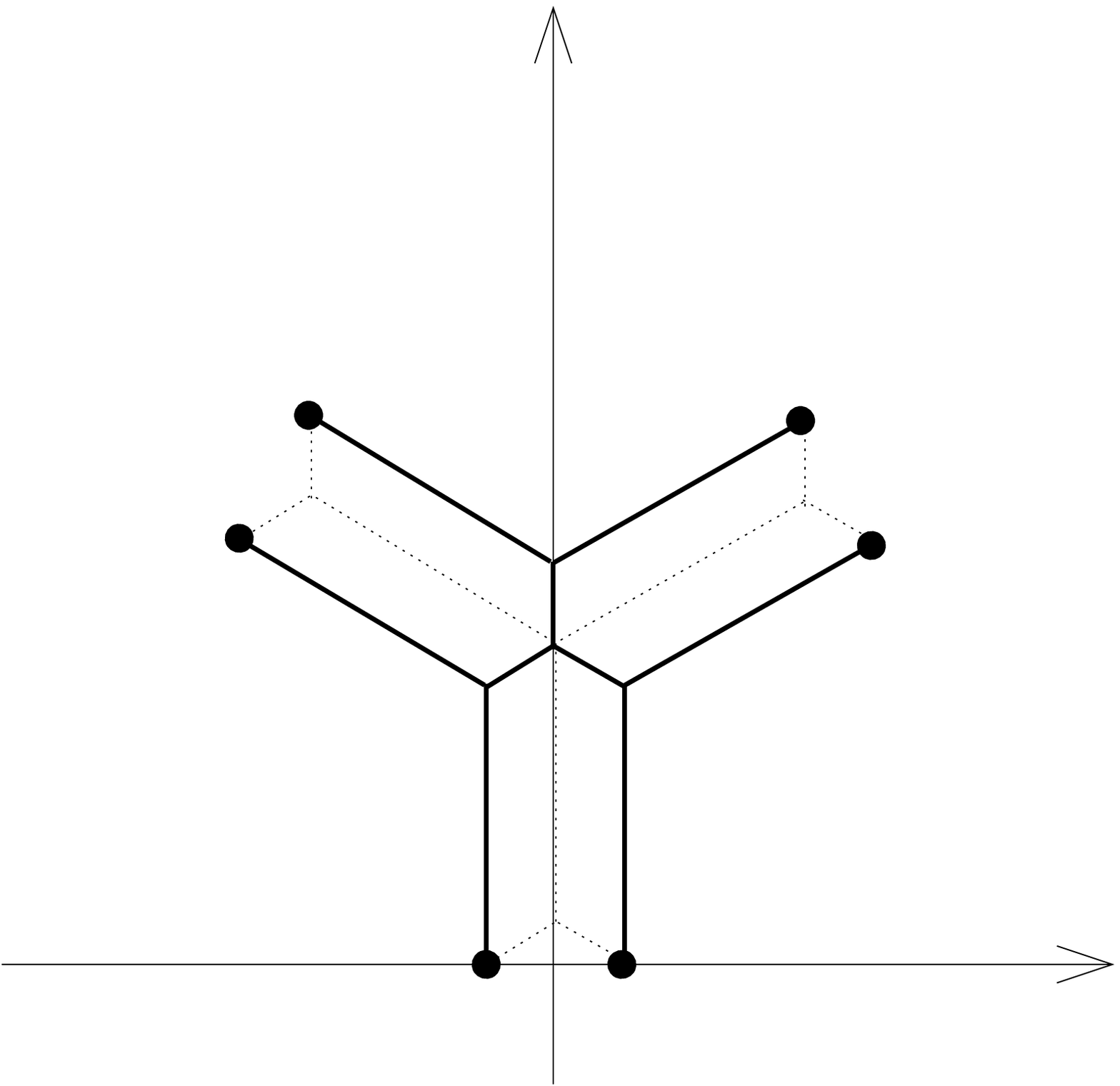}\\

(e)\hfil\hfil(f)\\

\caption{$\P$ (left) and $\Q$ (right) for increasing values of $s$.}
\label{comprs}
\end{figure}

We must also consider that the soap film will not always reach an equilibrium state by just detaching from an initial configuration. This is indeed the first step but afterwards some Steiner points can collide, so that the film will change its topology and then keep on moving towards an equilibrium state. Fig.~\ref{mov}(a) shows the WSMT for the vertices $A$, $B$, $C$, $D$ of a rectangle $\cR$ where $AB=CD=2$ and $BC=\ell$, $\frac{2}{\sqrt{3}}<\ell\le2$. The weight function is given by $w(A)=w(D)=\omega\ge7$ and $w(B)=w(C)=1$.

\begin{figure}[ht!]
\center
\includegraphics[scale=.51]{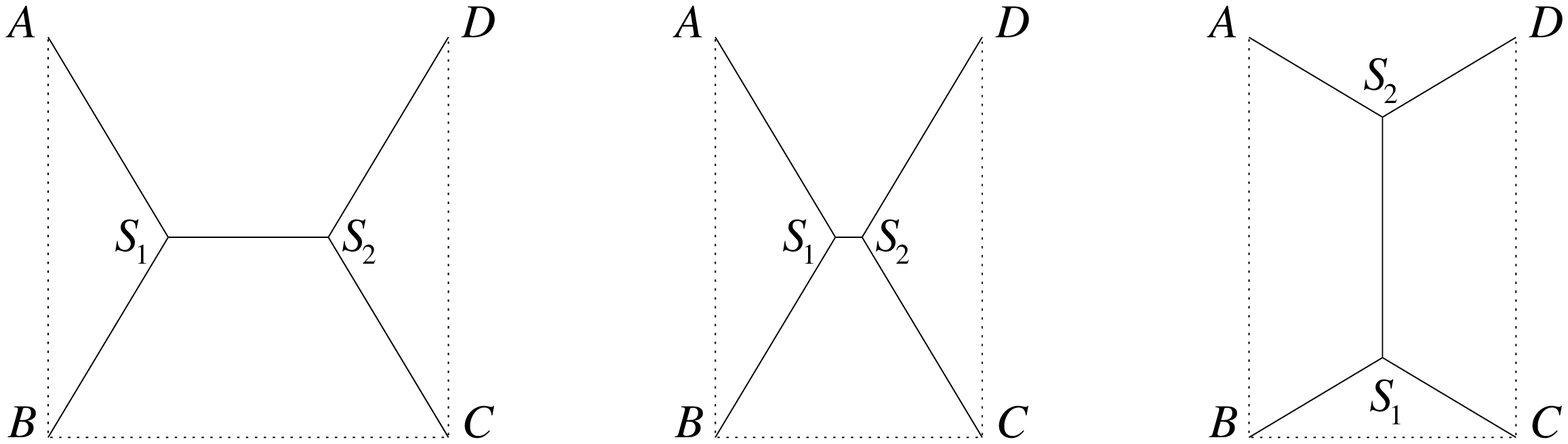}

(a)\hfil\hfil(b)\hfil\hfil(c)\\

\caption{STs for (a) $\ell=2$; (b) $\ell\gtrsim\frac{2}{\sqrt{3}}$ with initial topology; (c) same $\ell$ with another topology.}
\label{mov}
\end{figure}

Since $\ell>8/7$ the WMST is $\cR\setminus AD$, though either $\cR\setminus AB$ or $\cR\setminus CD$ stand for the MST. Hence the soap film detaches from $\cR\setminus AD$, as depicted in Fig.~\ref{mov}(b). But in our setting $w(S_1)=w(S_2)=1$, and therefore it is Fig.~\ref{mov}(c) that shows the WSMT, which is also the SMT.

Of course, the soap film in Fig.~\ref{mov}(b) will change to the one in Fig.~\ref{mov}(c) if we softly blow in the direction of $S_1S_2$. Our algorithm considers theses changes under the following:

\begin{asspt}
Let $\T=(V',E')$ be the detachment from $G=(V,E)$, in which there exists $S_1,S_2\in V'\setminus V$, $S_1S_2\in E'$ with $\frac{S_1S_2}{PS_j}<\sqrt{3}-1$ for at least three segments $PS_j\in E',\,j=1,2$. Then $\T$ is not the WSMT and so its topology must be changed in order to reduce $||\T||$.
\label{asspt1}
\end{asspt}

The upper bound $\sqrt{3}-1$ is attained in Fig.~\ref{mov}(a). Notice the similarity between Result~\cite[\S8.4]{GP} and Assumption~\ref{asspt1}, and also that we have a sufficient but not necessary condition. As an example, $\T$ will cease to be an SMT in Fig.~\ref{comprs}(d) for $s\gtrsim2-\sqrt{3}$ but there we have $\frac{S_1S_2}{PS_j}\ge1$ until $s=\frac{1}{2}-\frac{\sqrt{3}}{9}$.

Evolver seeks for minimisation, hence as a first step we are going to have $S_1=S_2$ in Fig.~\ref{mov}(b) for any $\ell\le\frac{2}{\sqrt{3}}$. But the case $S_1S_2=0$ is comprised by Assumption~\ref{asspt1} and therefore we do not need to implement changes of topology by collision of Steiner points separately.

In Assumption~\ref{asspt1} we stated ``at least three segments'' because of the following example: suppose $A$ were much closer to $S_1$ in Fig.~\ref{mov}(b) so that $\frac{S_1S_2}{AS_1}\ge\sqrt{3}-1$ by keeping the same topology. Now softly blowing $S_2$ against $S_1$ will make this latter collide with $A$, and the new topology is depicted in Fig.~\ref{moving}(a). According to our convention $w(A)$ will take the same value as $w(S_1)$, and therefore we reduce $||\T||$. 

\begin{figure}[ht!]
\center
\includegraphics[scale=.53]{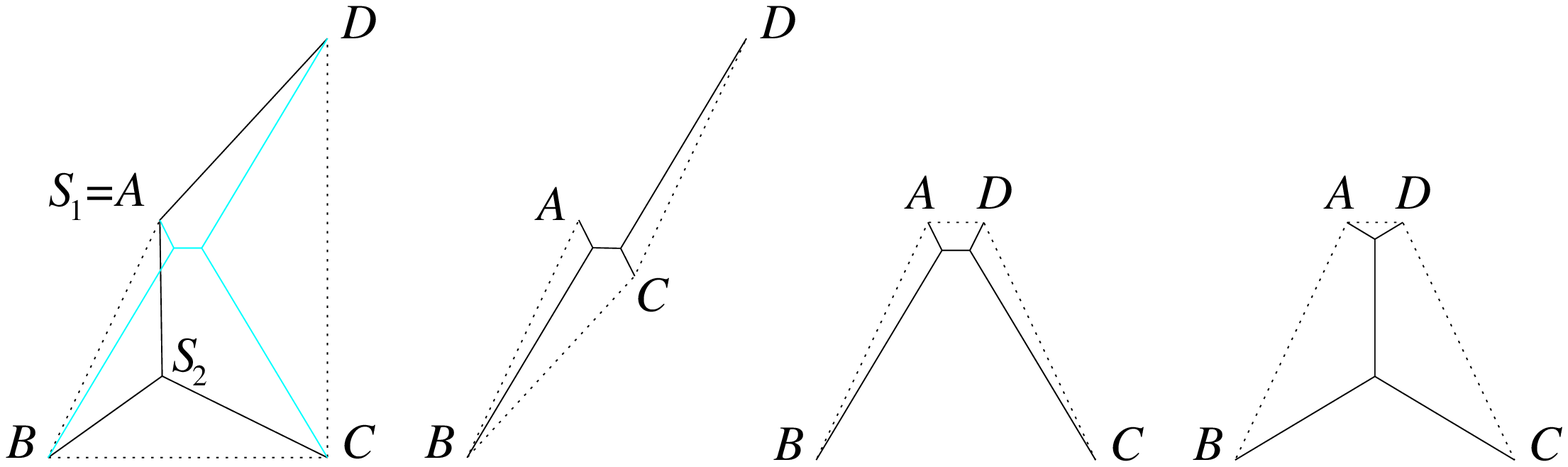}

(a)\hfil\hfil(b)\hfil\hfil(c)\hfil\hfil(d)\\

\caption{Towards the WSMT with (a) script-intervention; (b) {\tt g}-automatic; (c,d) without intervention.}
\label{moving}
\end{figure}

This would not work if we had stated ``at least two segments'', as illustrated in Fig.~\ref{moving}(b). There the segment $S_1S_2$ is tilted clockwise by an imperceptible angle of $1.3^\circ$, which Evolver automatically does with its {\tt g}-command (see \cite{B} for details). A possible improvement in Assumption~\ref{asspt1} could be about the position of these two segments: take $\ell=\delta+2/\sqrt{3}$, $A$, $D$ sufficiently close to $S_1$, $S_2$ and $\delta\cong0$. Then Figs.~\ref{moving}(c,d) depict the WMST in dotted line, so that $||\T||$ can be reduced by the same ``blowing'' if we use the given weights. However, with $\omega=7$ consider that $w(B)$ and $w(C)$ change to $8$ and $6$, respectively. Then first $S_1$ detaches with weight $7$, then afterwards $S_2$ with weight $6$. In this case Fig.~\ref{moving}(d) still leads to a reduction in $||\T||$ but not as expressive as for $w(S_2)=w(S_1)=6$. 

Namely, the sequence of detachment should be swapped, which in practice changes the way the plates are taken out from the soapy water. However, it is pointless to strengthen Assumption~\ref{asspt1} since we have no formal proof of its formulation yet. Therefore, in this present version we have decided to skip this improvement, which searches for a WSMT, and content ourselves with reproducing the physical experiment with less intervention of blowing to reduce the total length of the soap film.

The next two sections explain that we always work with plane graphs. The reader could be curious about the algorithm giving a WSMT with intersections, so that they should be further treated. For instance, suppose we could attribute weights to the bullets in Fig.~\ref{nev} such that the resulting WMST is marked there in black with the detachments in blue. 

\begin{figure}[ht!]
\center
\includegraphics[scale=.25]{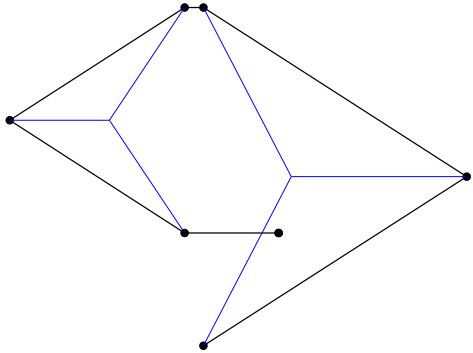}

\caption{Two mirrored equilateral triangles with resizing and slightly apart.}
\label{nev}
\end{figure}

However, according to our tests apparently there is no function $w:V\to\R_+^*$ in Definition~\ref{def1} that produces a WMST as in Fig.~\ref{nev}. Furthermore, in Sect.~\ref{res} we discuss several tests with our algorithm, for arbitrary random number and position of vertices with random $w$, and curiously no intersections ever occur to the heuristic WSMT. Therefore we consider:

\begin{asspt}
The heuristic WSMT obtained through our method is always plane.
\label{asspt2}
\end{asspt}

Both Assumptions~\ref{asspt1} and~\ref{asspt2} still lack a formal proof, though they have always shown to be true in our simulations.

\section{Limitations of the Heuristic}
\label{limh}

Our algorithm aims at reproducing a physical experiment and also direct it in order to find the WSMT for certain weights attributed to a set of vertices. In Sect.~\ref{meth} we shall explain how the user can enter points according to a sequence that will determine the way the plates are removed from the soap solution. This can also be defined by weights but one can give priority to finding the WSMT, which is the main purpose of our algorithm. The sought after soap film with the same configuration will in general need some intervention as explained in Sects.~\ref{prelim} and~\ref{backg}.

There we saw that $\P\cup V_1V_6$ is a regular hexagon for $s=1/3$, and $\P$ is also the SMT with $|\P|=10/3$. Small variations in $V$ can drastically change the topology of the SMT. For instance, if $V_1$ is rotated inwards by $10^\circ$ we get $\P$ and $\Q$ as depicted in Figs.~\ref{perturb}(a,b), respectively. There we have $|\P|\cong3.325$ and $|\Q|\cong3.399$, so that $\P$ is still the SMT. Our method is based on detachment of Steiner points from an initial configuration, and in this example we see that it works. But here of course the heuristic fails for $s\cong0$.

\bigskip

\begin{figure}[!htb]
\centering
\begin{minipage}{.5\textwidth}
\flushleft
\includegraphics[scale=0.31]{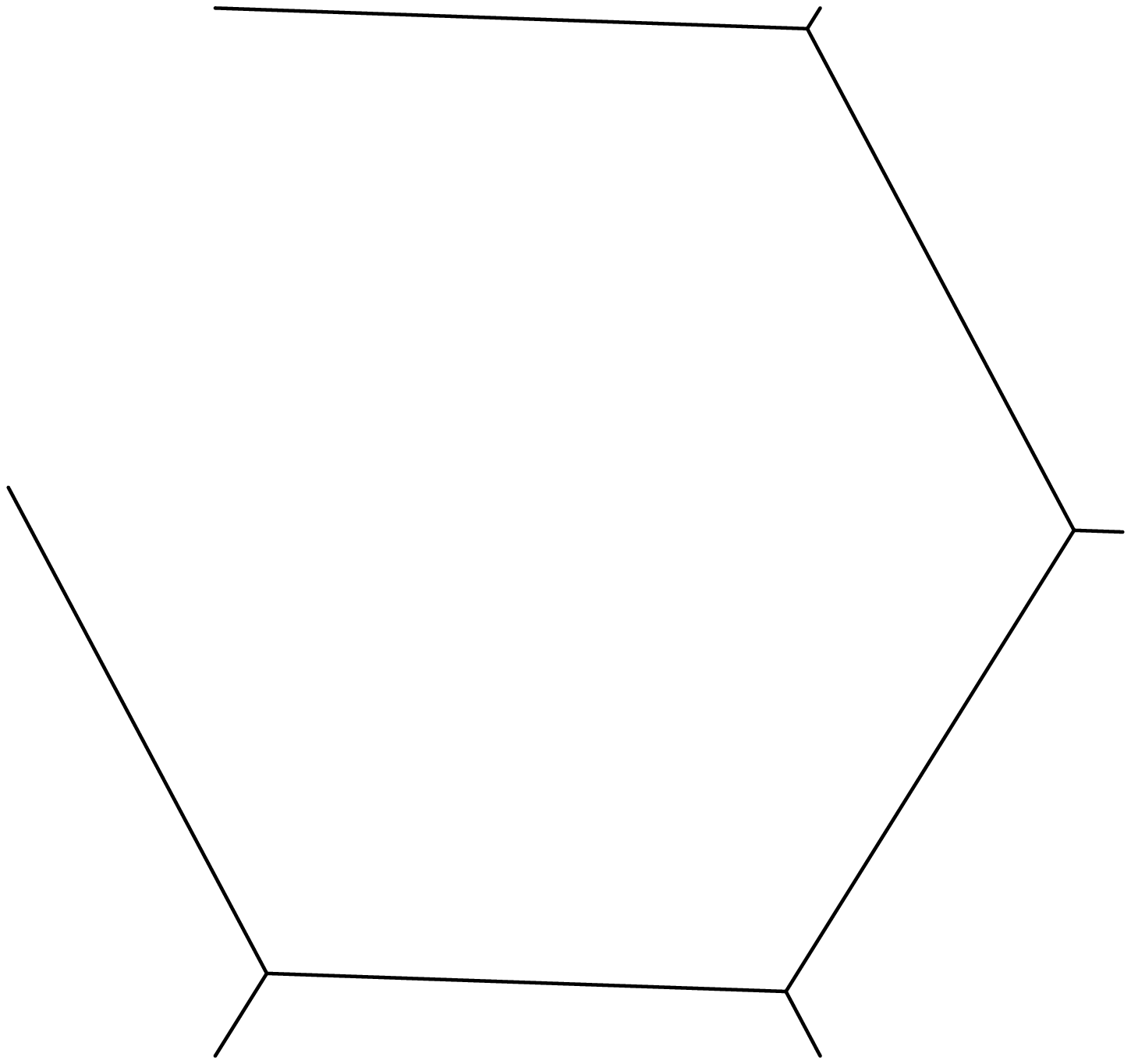}
\end{minipage}%
\begin{minipage}{.5\textwidth}
\flushright
\includegraphics[scale=0.31]{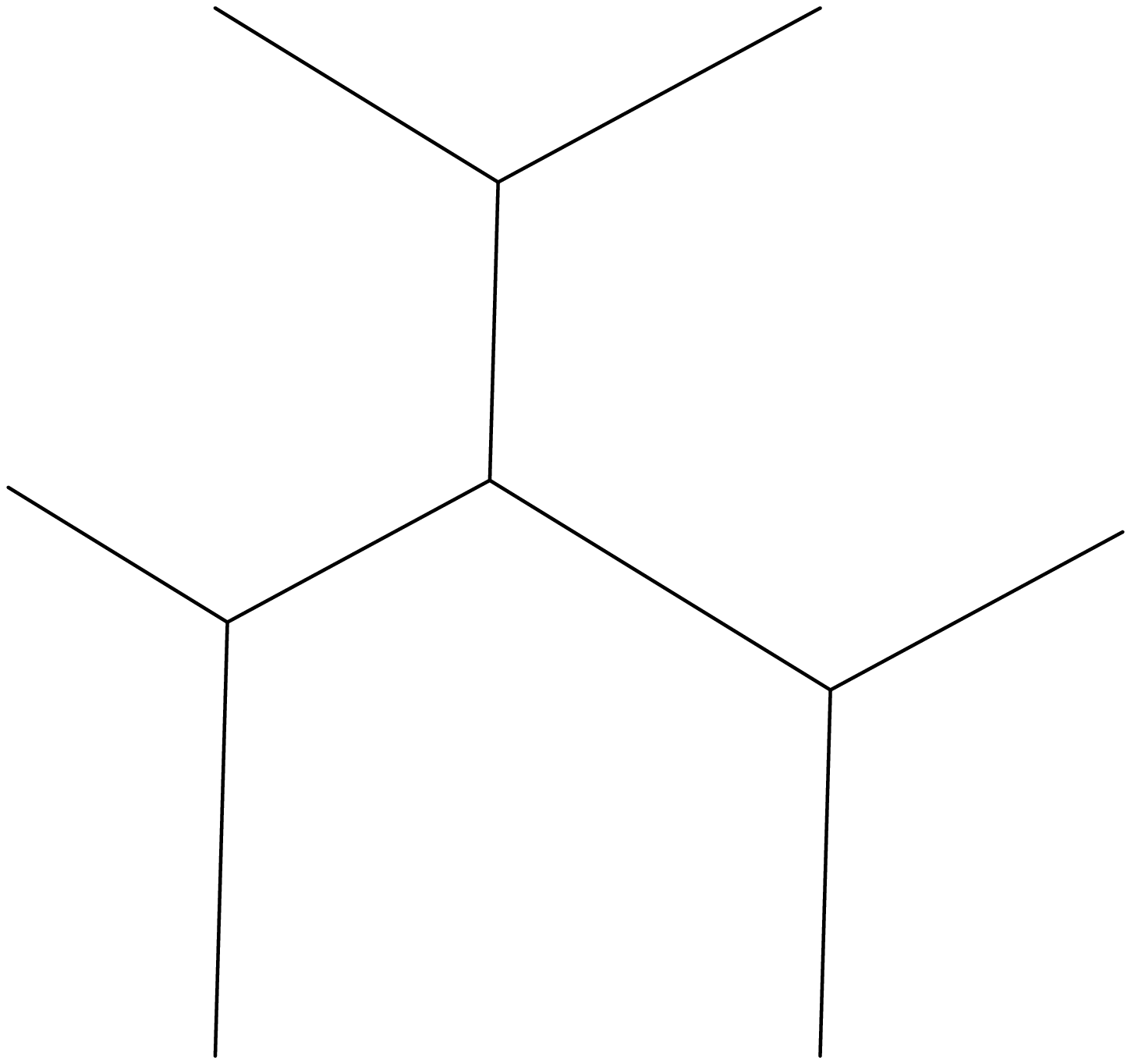}
\end{minipage}

\bigskip

(a)\hfil\hfil(b)\\

\caption{ST after rotating $V_1$ inwards by $10^\circ$ at $s=1/3$ for (a) $\P$ and (b) $\Q$.}
\label{perturb}
\end{figure}

Another limitation is that our algorithm is sequential, whereas slidings and detachments as in Figs.~\ref{sliding} and~\ref{mar}(b) all occur simultaneously, and also together with the action of taking out the plates from the soapy water. In particular, when this is done two film strips never cross at any moment.

Our program uses a preprocessing that gives a {\it plane} WMST, while the real WMST does allow crossings as we shall see in Sect.\ref{meth}. By turning such crossings into Steiner points the resulting heuristical WSMT would surely be shorter than the one obtained from the plane WMST. But according to our analysis the cost of implementing this approach is too high for little benefit, even if considering that the Gilbert-Pollak conjecture does not hold for the weighted case.

\section{Method}
\label{meth}

Now we explain our method with an example. The first part consists in giving a set of points with coordinates $(x,y)$ inside the square $[0,100]^2$, and each will have an integer weight $w\in[1,77]$. Hence the user must re-scale points to that range and render them either graphically or by means of a datafile. After invoking Octave type {\tt wmst} at the prompt and you will get the message
\\

{\tt Please adjust terminal window to show full picture.}

{\tt Give a filename to open or press Enter to choose points.}
\ \\

A graphical window appears by pressing the Enter key, and wherever you hover the mouse cursor a coordinate pair $(x,y)$ can be marked by hitting any character key. The corresponding weight is given by its ASCII number {\it minus} 48. Then the character keys {\tt 1} to {\tt 9} really mean these numbers as weights, whereas other keys will give weights up to 77 (for the character \#125 = {\tt\}}). Either mouse button or characters under \#49 are all set to $w=1$, and non-integer weights must be edited in datafile. Namely, you can save your weighted points and change them manually afterwards. Please see \url{https://theasciicode.com.ar} for the ASCII numbers.

The Enter key terminates your graphical input and readily gives the plane WMST, though the actual tree is in general non-embedded. Fig~\ref{wmst}(a) shows the set of terminals in {\tt sample.txt}, which you can also input to our program (without the extension {\tt txt}). The total weighted length is $5199.7$ for that WMST. Intuitively speaking, Steiner points should replace intersections but these can happen in various ways: multiple at different and same points, between a vertex and an inner point of an edge, etc. Treating all these cases geometrically will increase computational cost for little use because any heuristic already has limitations. In our present version we chose to get $\P$, the plane WMST, by allowing a greater connection cost. Namely, our adaptation of Prim's algorithm looks for edges of lowest cost that do not intersect the already existing ones. In our example Fig.~\ref{wmst}(b) shows $\P$ with $||\P||=5498.4$ and total Euclidean length $|\P|=589.2$ (if we change the weight function to $w\equiv 1$). There some angles are marked in red and magenta for the longer and shorter sides, respectively. These will start the sliding process depicted in Fig~\ref{sliding}.

\begin{figure}[!htb]
\centering
\begin{minipage}{.5\textwidth}
\flushleft
\includegraphics[scale=0.51]{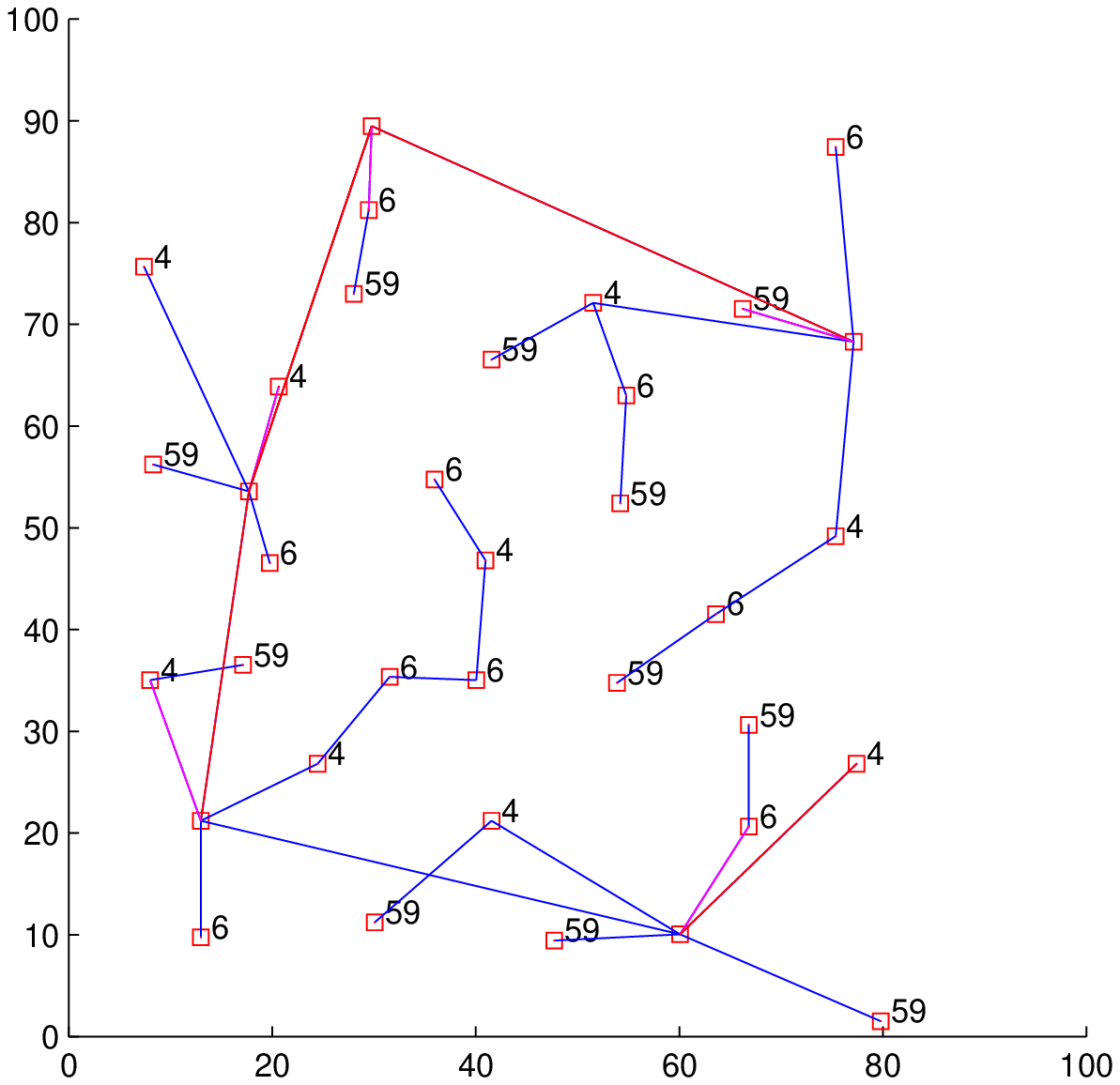}
\end{minipage}%
\begin{minipage}{.5\textwidth}
\flushright
\includegraphics[scale=0.71]{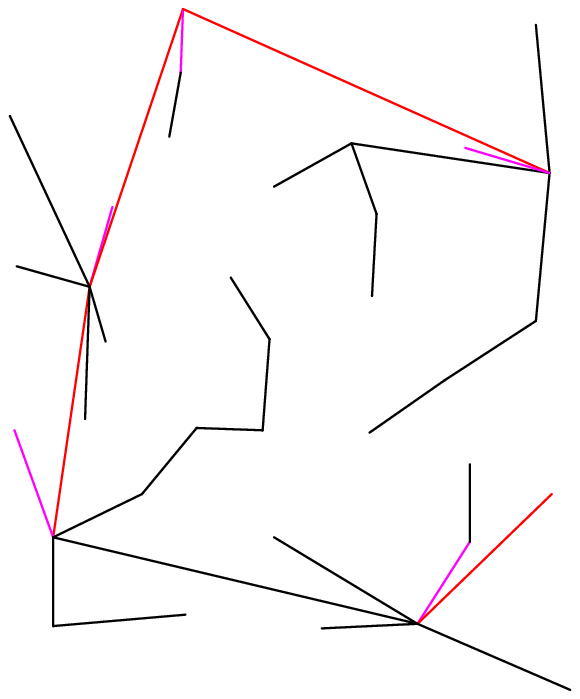}
\end{minipage}

(a)\hfil\hfil(b)\\

\caption{Terminal points and (a) the non-embedded WMST; (b) the plane WMST.}
\label{wmst}
\end{figure}

The Octave program {\tt wmst} will store the plane WMST in the Evolver file {\tt st.fe}, and now we must invoke it at the shell prompt by typing {\tt evolver h}, where {\tt h} stands for {\it heuristic}. Its first step is depicted in Fig~\ref{wmst}(b), and from that point on each vertex $B$ with acutest angle $A\widehat{B}C$ is checked to see whether $\max\{B\widehat{C}A,C\widehat{A}B\}\ge120^\circ$. If so, then either $A$ or $C$ is an intrinsic Steiner point, and the red side is replaced with the missing edge of $\Delta ABC$. This process is repeated until exhausting the number of intrinsic Steiner points. Figs.~\ref{sl_bg_end}(a,b) show the first and last iterations, respectively. 

\begin{figure}[!htb]
\centering
\begin{minipage}{.5\textwidth}
\flushleft
\includegraphics[scale=0.37]{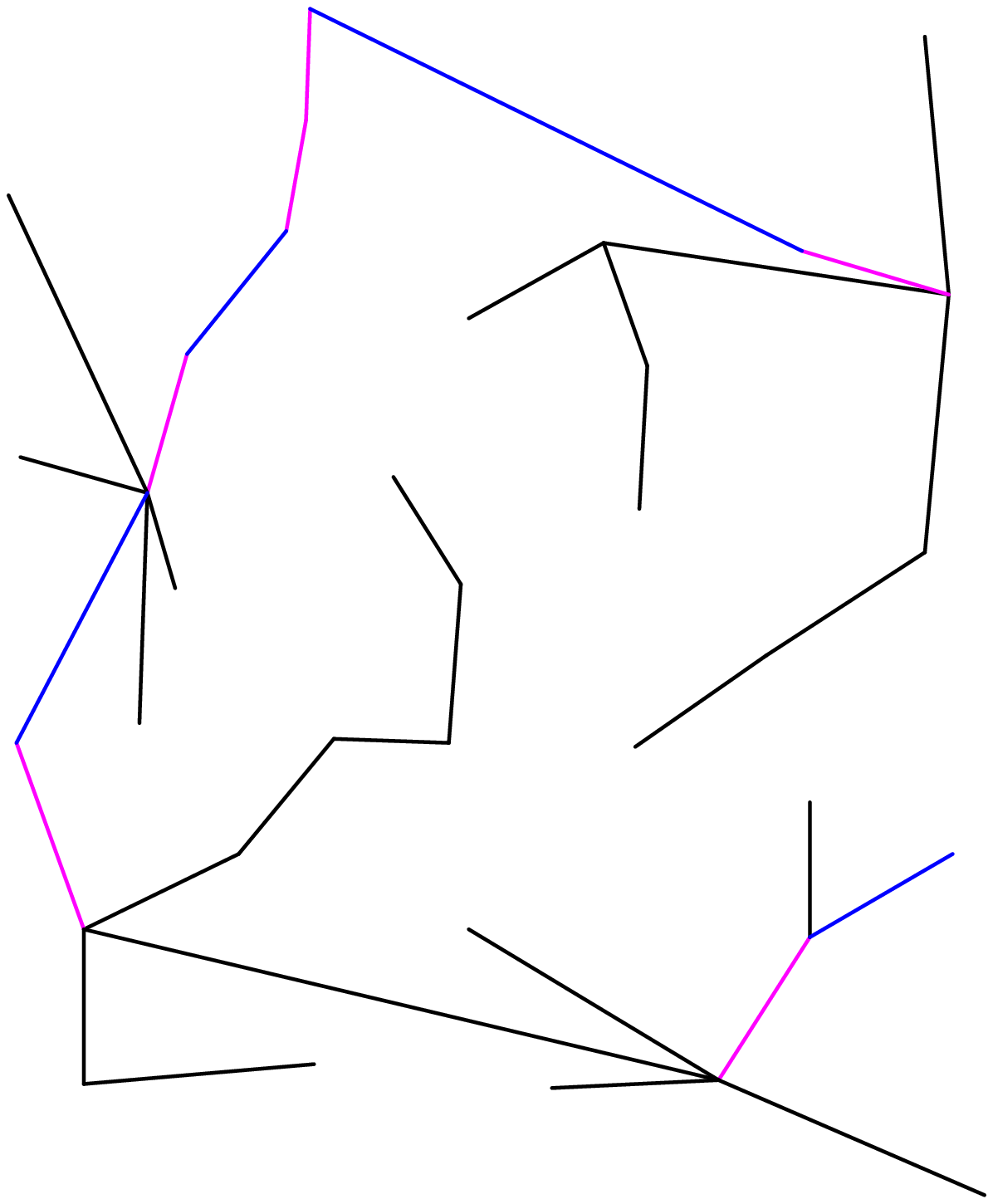}
\end{minipage}%
\begin{minipage}{.5\textwidth}
\flushright
\includegraphics[scale=0.37]{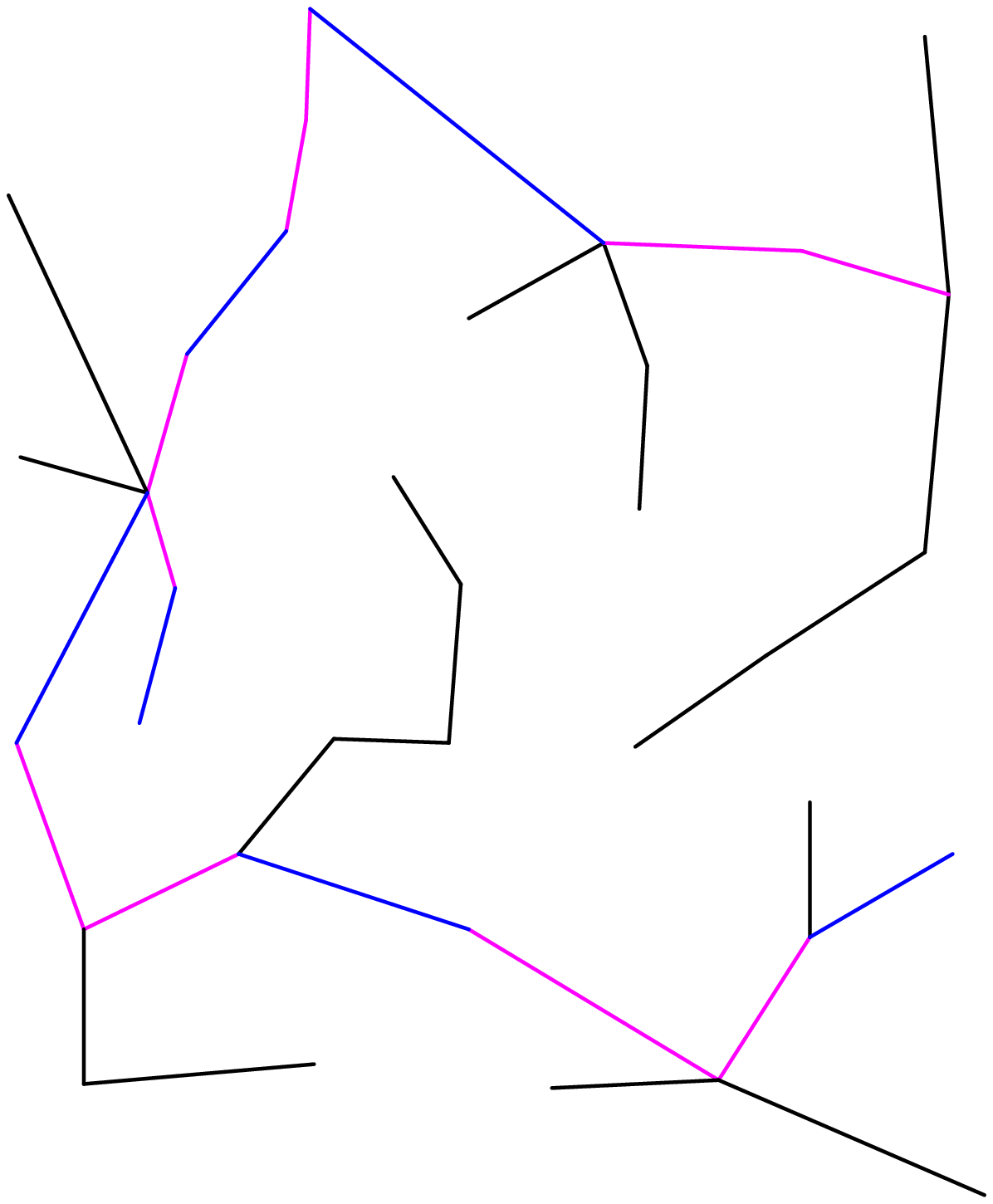}
\end{minipage}

(a)\hfil\hfil(b)\\

\caption{Sliding process at (a) the beginning; (b) the conclusion.}
\label{sl_bg_end}
\end{figure}

In this example the sliding has not started from the most acute angle in Fig.~\ref{wmst}(b). As explained in Sect.~\ref{limh} the soap film configuration will always depend on the way the plates are removed from the soapy water. Actually vertices are checked in the order they were entered by the user.

The next step is to detach Steiner points according to Sect.~\ref{prelim}. In practice this happens simultaneously with the sliding process but our approach is sequential. Fig.~\ref{detach}(a) show the first iteration of detachment, and there we left the replaced edges in cyan for visualisation. The last iteration is shown in Fig.~\ref{detach}(b), where we marked in red the segments $S_1S_2\in E'$ to see whether they satisfy the hypothesis of Assumption~\ref{asspt1}. Finally in Fig.~\ref{finalz} we see the last steps to get the soap film.

\begin{figure}[!htb]
\centering
\begin{minipage}{.5\textwidth}
\flushleft
\includegraphics[scale=0.37]{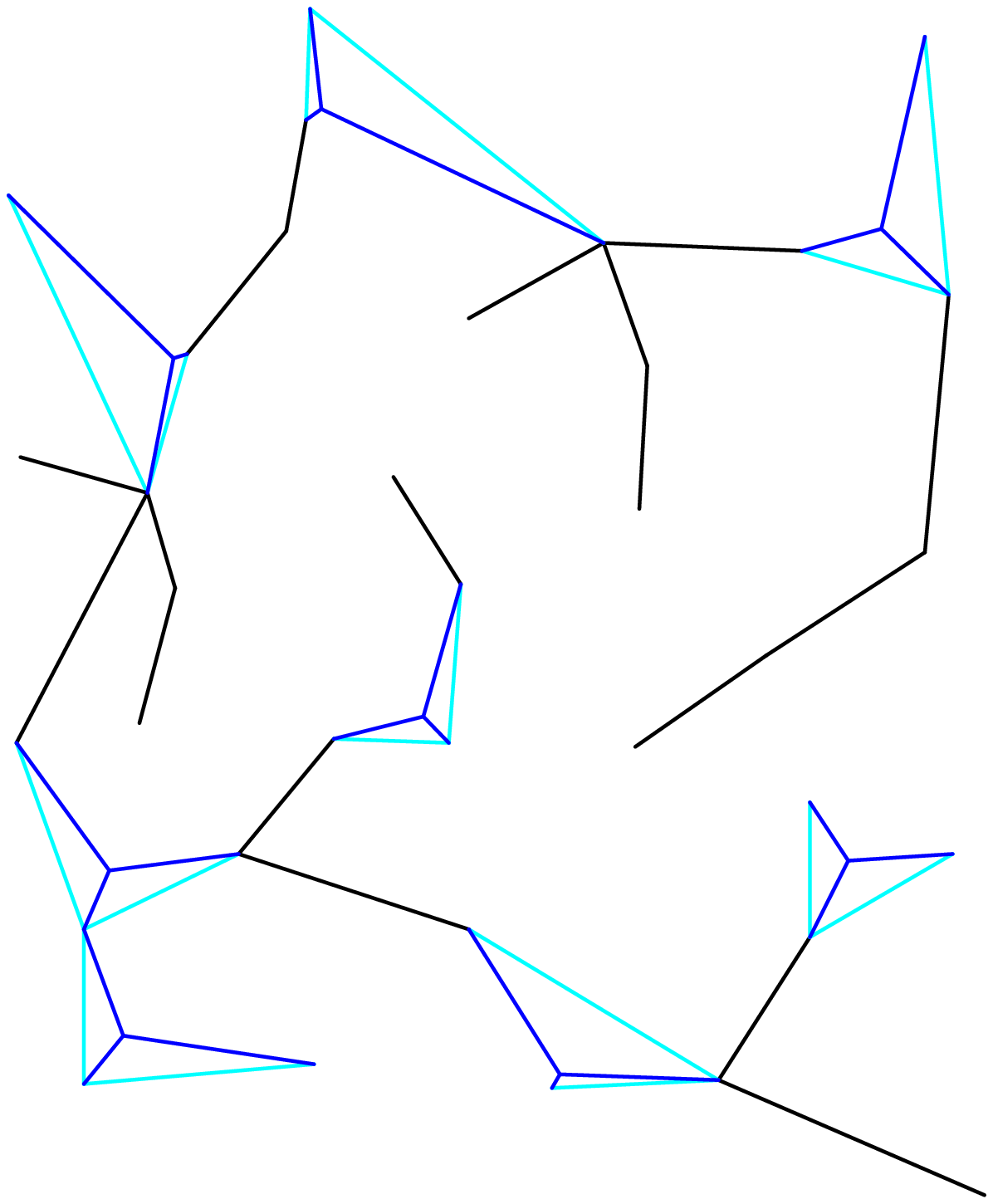}
\end{minipage}%
\begin{minipage}{.5\textwidth}
\flushright
\includegraphics[scale=0.37]{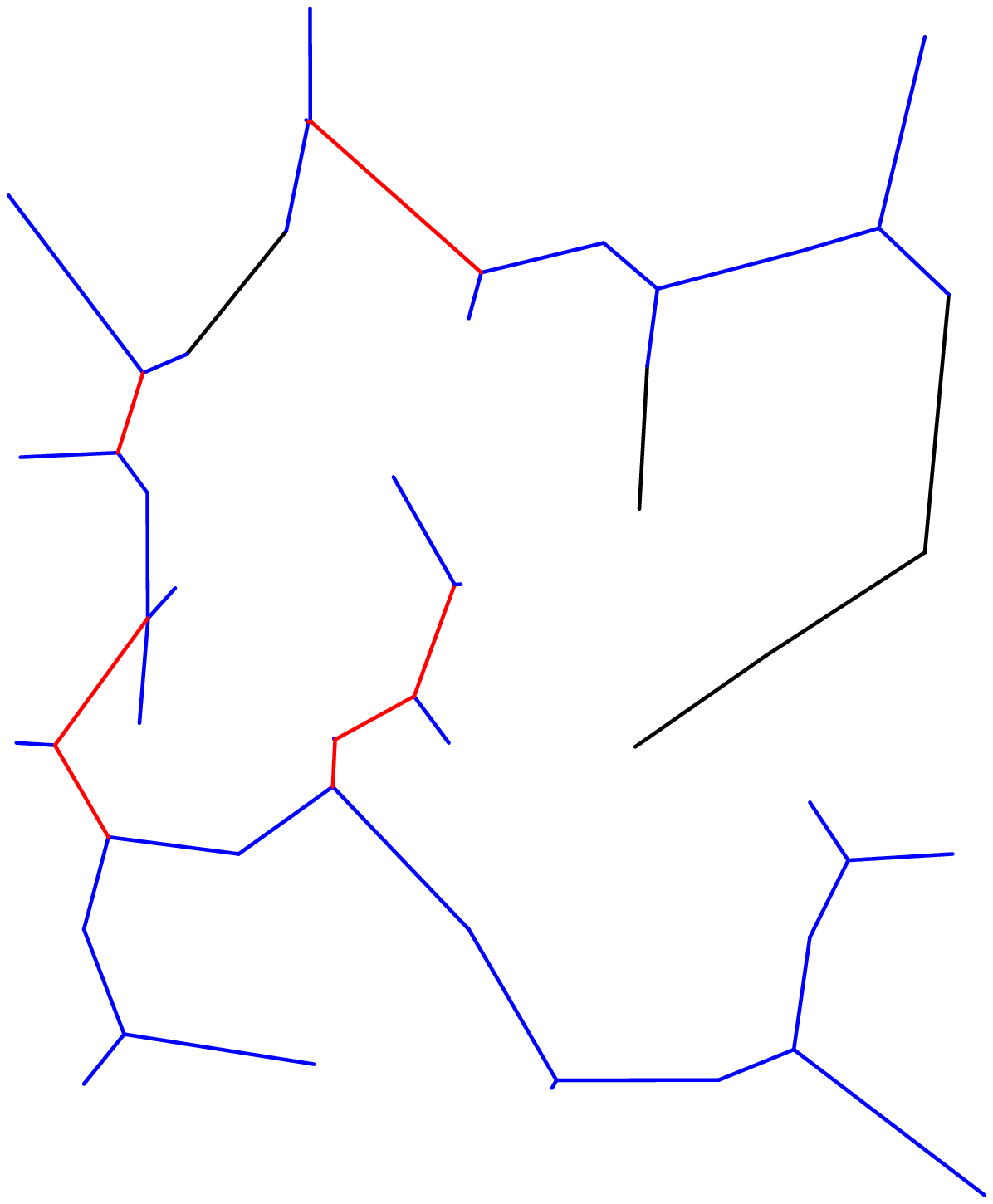}
\end{minipage}

(a)\hfil\hfil(b)\\

\caption{Detachment from Steiner points at (a) the beginning; (b) the conclusion.}
\label{detach}
\end{figure}

Fig.~\ref{finalz}(b) shows $\T$, our heuristic WSMT, in black and also $\P$ in cyan for comparisons. The program gives $||\T||=3065.6$ and $|\T|=406.6$, with ratios $||\T||/||\P||=0.56$ and $|\T|/|\P|=0.69$, namely both below $\sqrt{3}/2$.

\begin{figure}[!htb]
\centering
\begin{minipage}{.5\textwidth}
\flushleft
\includegraphics[scale=0.37]{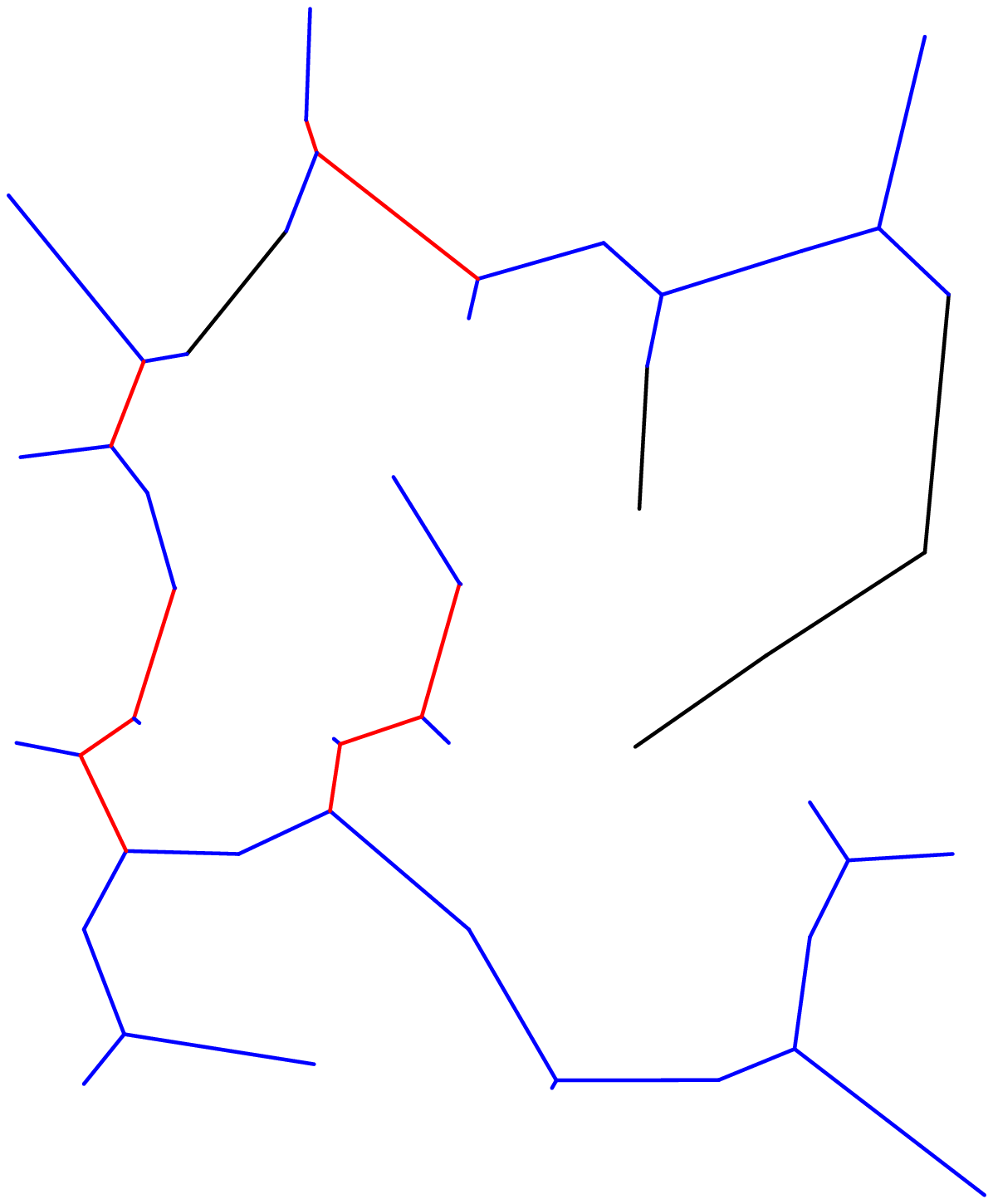}
\end{minipage}%
\begin{minipage}{.5\textwidth}
\flushright
\includegraphics[scale=0.37]{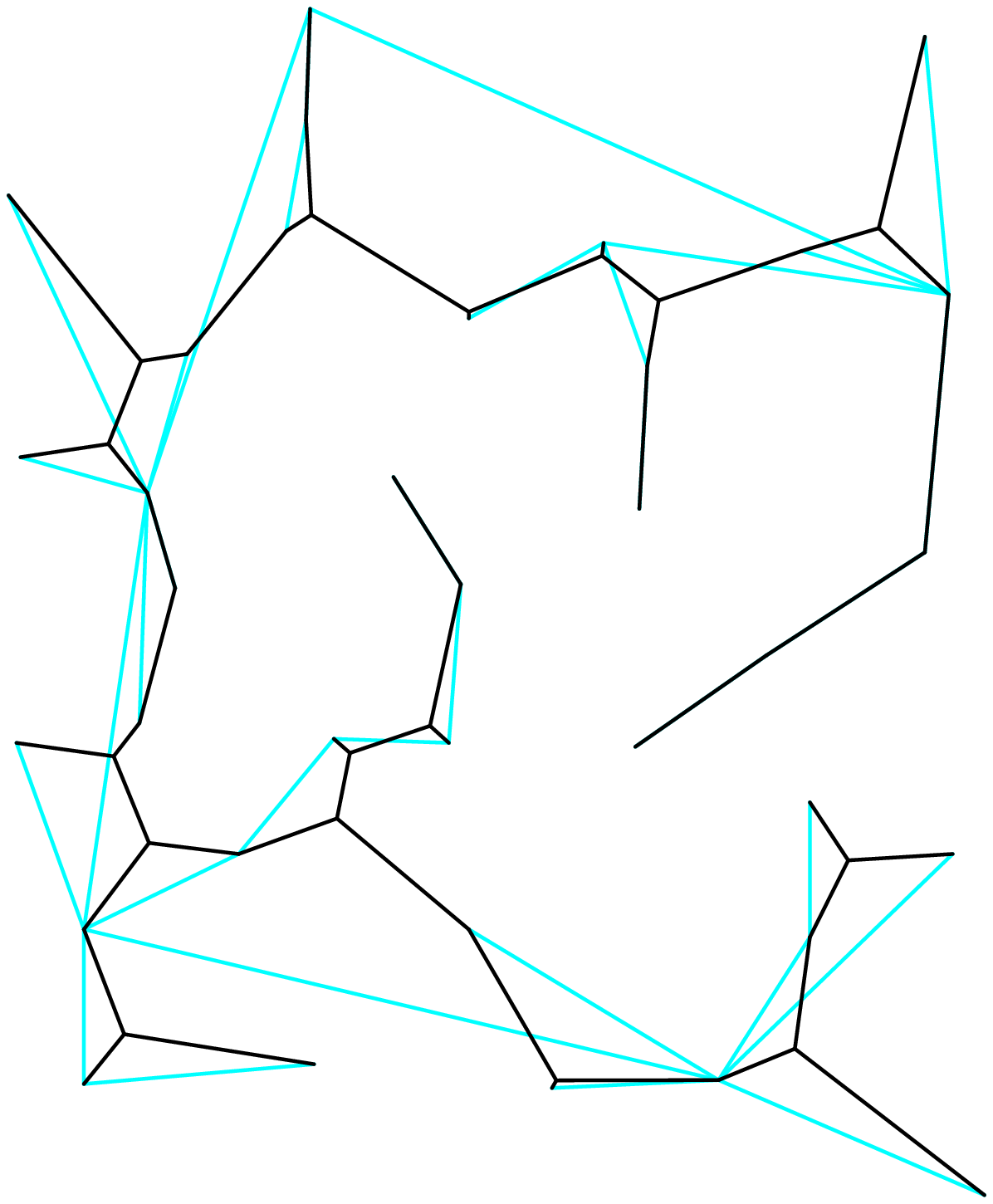}
\end{minipage}

(a)\hfil\hfil(b)\\

\caption{The soap film (a) after applying Assumption~\ref{asspt1}; (b) compared with the WMST.}
\label{finalz}
\end{figure}

However, Fig.~\ref{finalz}(b) illustrates our program when it stops at a tolerance of 2.2\% regarding the theoretical 120$^\circ$. Namely, in an ideal ST adjacent edges must always make an angle of at least 120$^\circ$ but the closer we want to approach it the longer becomes the computational time. Fig.~\ref{jump}(a) shows that there is little change when we set the tolerance to 1.1\% but now $||\T||$ jumps to 3775.8 (all other values remain the same).

\begin{figure}[!htb]
\centering
\begin{minipage}{.5\textwidth}
\flushleft
\includegraphics[scale=0.37]{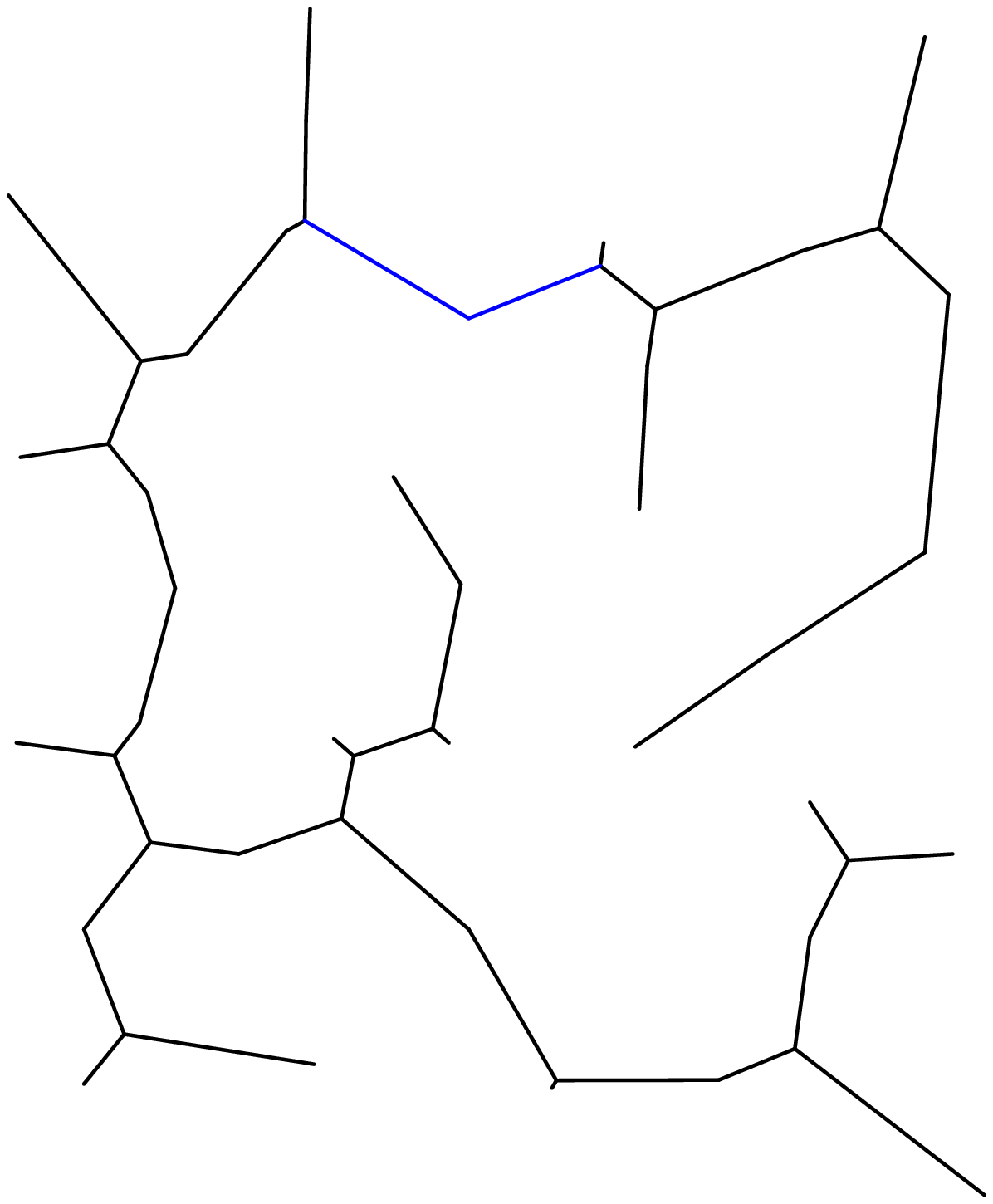}
\end{minipage}%
\begin{minipage}{.5\textwidth}
\flushright
\includegraphics[scale=1.00]{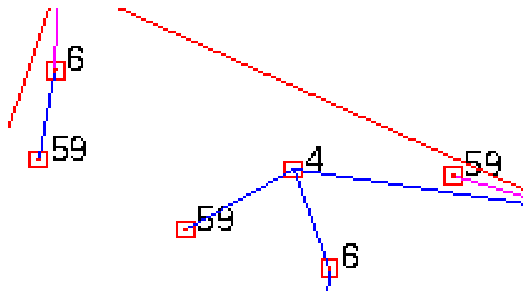}
\end{minipage}

(a)\hfil\hfil(b)\\

\caption{(a) 120$^\circ$ angles now with 1.1\% error margin; (b) zoomed detail of the initial WMST.}
\label{jump}
\end{figure}

The reason for this jump is explained by the blue segments in Fig.~\ref{jump}(a) and part of Fig.~\ref{wmst}(a) zoomed in Fig.~\ref{jump}(b). One of the Steiner points in Fig.~\ref{finalz}(b) now collided with a terminal that weighs 59, hence what was a local weighted connection of approximately $14.2(6+4)/2+10.5(4+4)/2=113$ now becomes $14.2(6+59)/2+10.5(59+4)/2=792.25$. Their difference of 679.25 is 95.7\% the actual rise of 710.2 in $||\T||$ (do not forget that nearby segments also changed and one Steiner point collapsed to a terminal). Notice that this rise is a significant 23.2\%.

The next section is devoted to some tests and comparisons performed with our algorithm.

\section{Results}
\label{res}

As already explained in Sect.~\ref{meth} our preprocessing with MATLAB/Octave obtains the plane WMST for a given set of vertices, and we gave an example in Fig.~\ref{wmst}. Indeed, if we improve our heuristic by working with the actual WMST this will increase the complexity of the algorithm for very little gain.

But Sect.~\ref{backg} shows an example for which intersections could happen in the heuristic WSMT, so that a postprocessing would then be necessary. However, in Computational Geometry we can perform tests in order to introduce assumptions under which the algorithm is implemented. This is very important because otherwise we take the risk of elaborating additional source code in vain.

As commented on Fig.~\ref{nev} our tests have never shown a WMST as depicted there, and some of the outputs are shown in Fig.~\ref{sixbirds}.

\begin{figure}[ht!]
\center
\includegraphics[scale=.22]{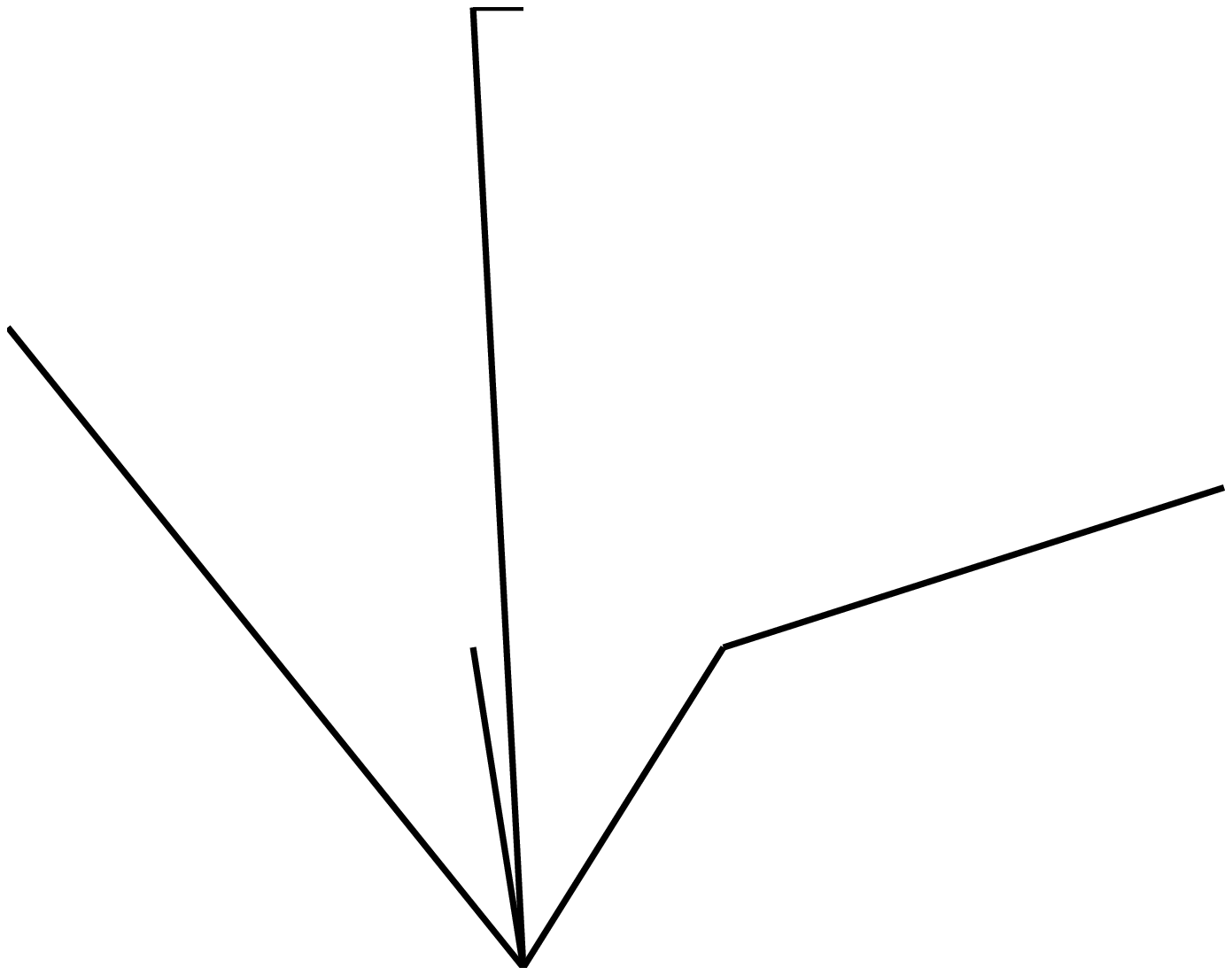}
\includegraphics[scale=.22]{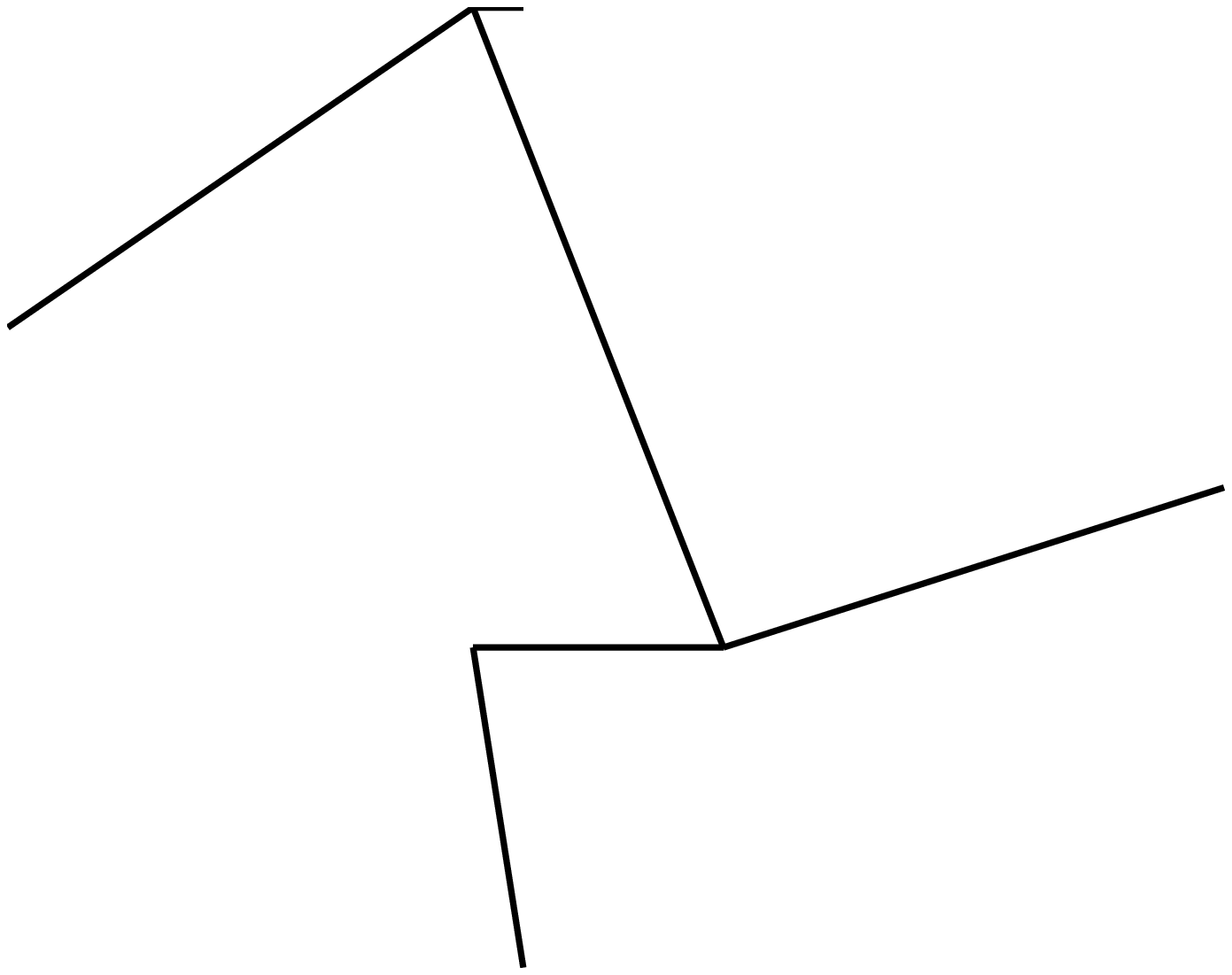}
\includegraphics[scale=.22]{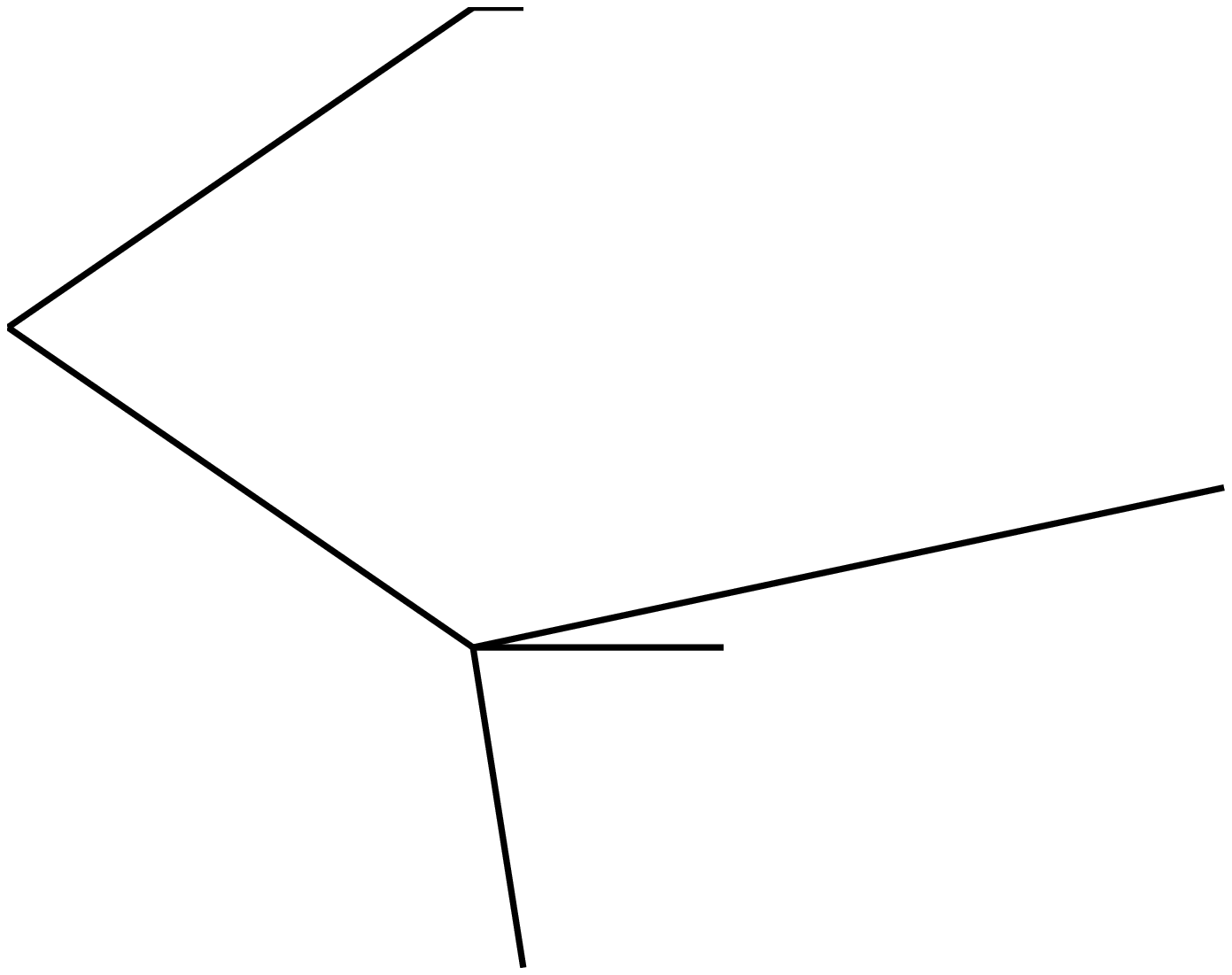}
\ \\
\includegraphics[scale=.22]{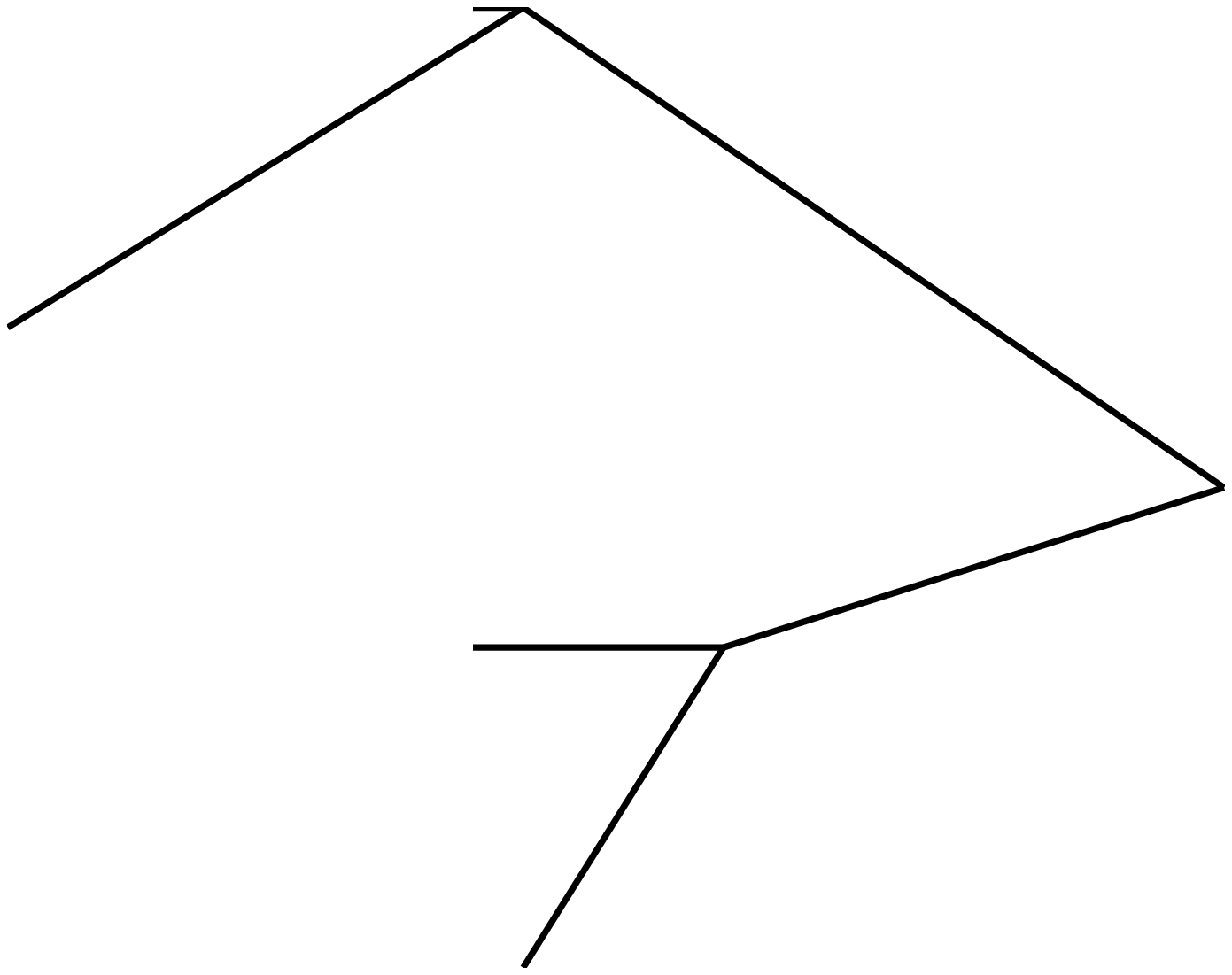}
\includegraphics[scale=.22]{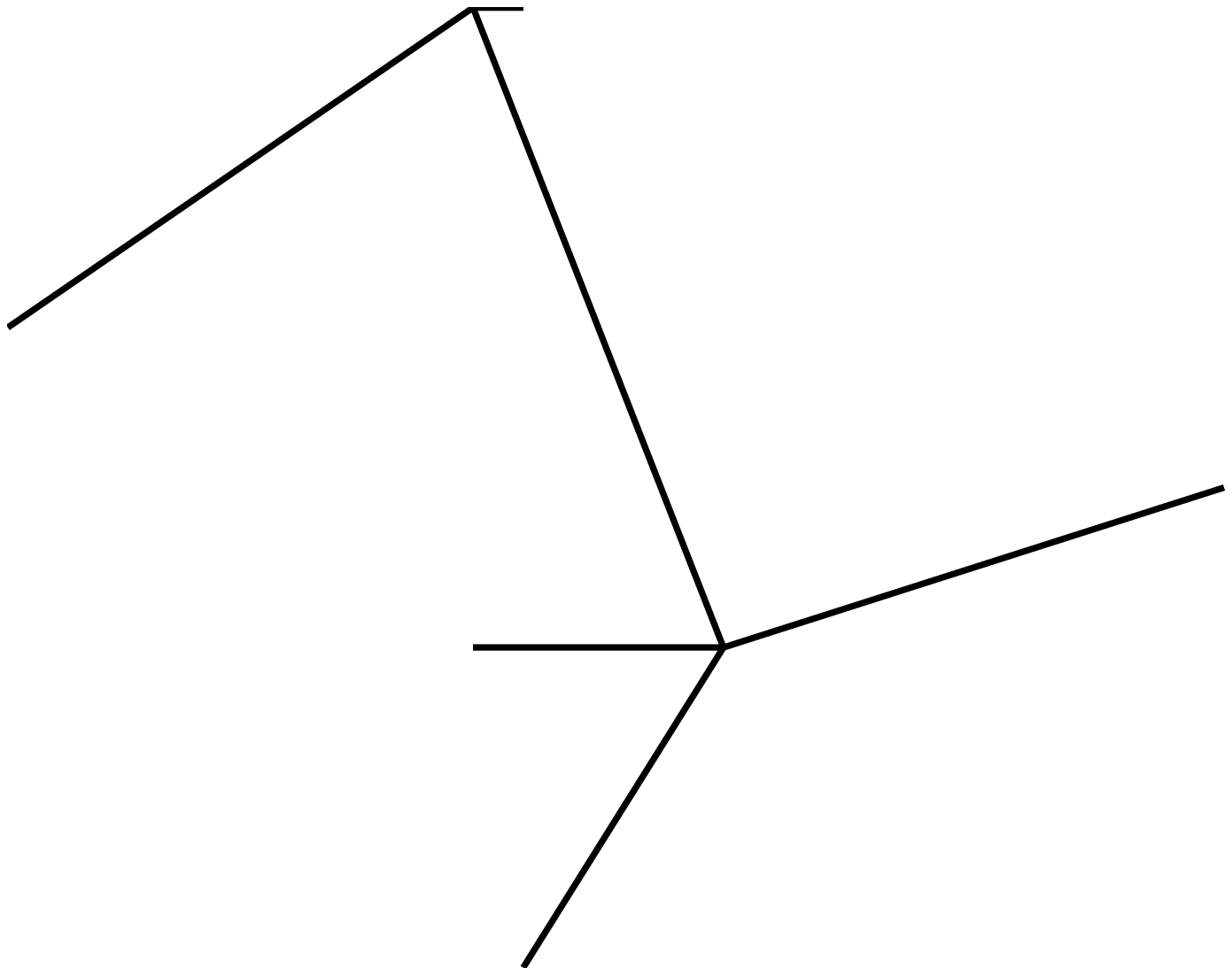}
\includegraphics[scale=.22]{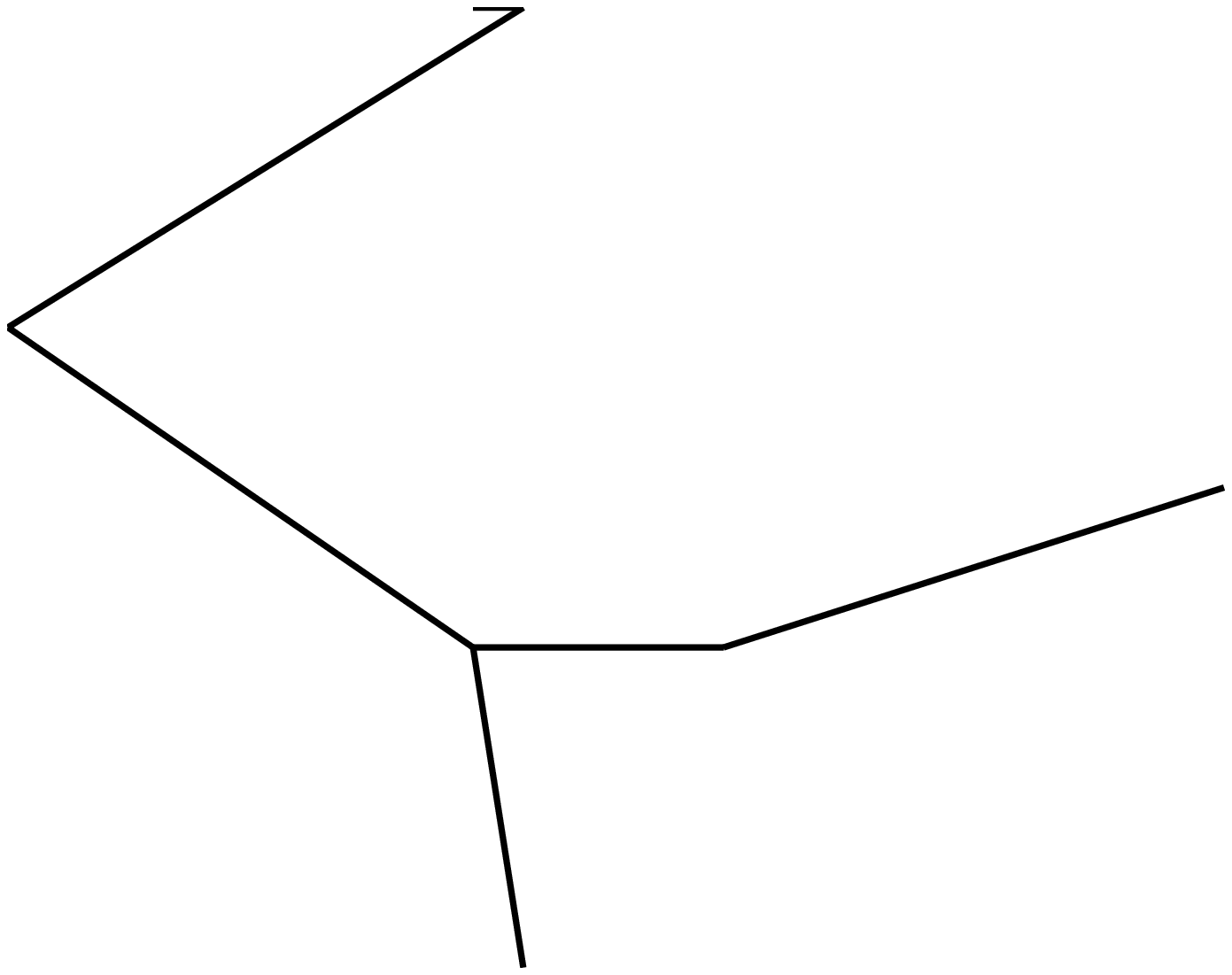}

\caption{Six of the several tests on WMST with the vertices in Fig.~\ref{nev}.}
\label{sixbirds}
\end{figure}

That is why we introduced Assumption~\ref{asspt2}, which one fine day will be either proved or disproved, and in this last case we shall complete what is missing in the algorithm. If the reader is willing to check some outputs of the preprocessing that generated Fig.~\ref{sixbirds}, here is the code in Fig.~\ref{fift}. Notice that its actual algorithm has only 15 lines, whereas the rest just stands for initialisation and formatting. Several pictures, say 12, come by simply entering the command {\tt for count=1:12 test;end} at the prompt.  

\begin{figure}[ht!]
\center
\includegraphics[scale=.33]{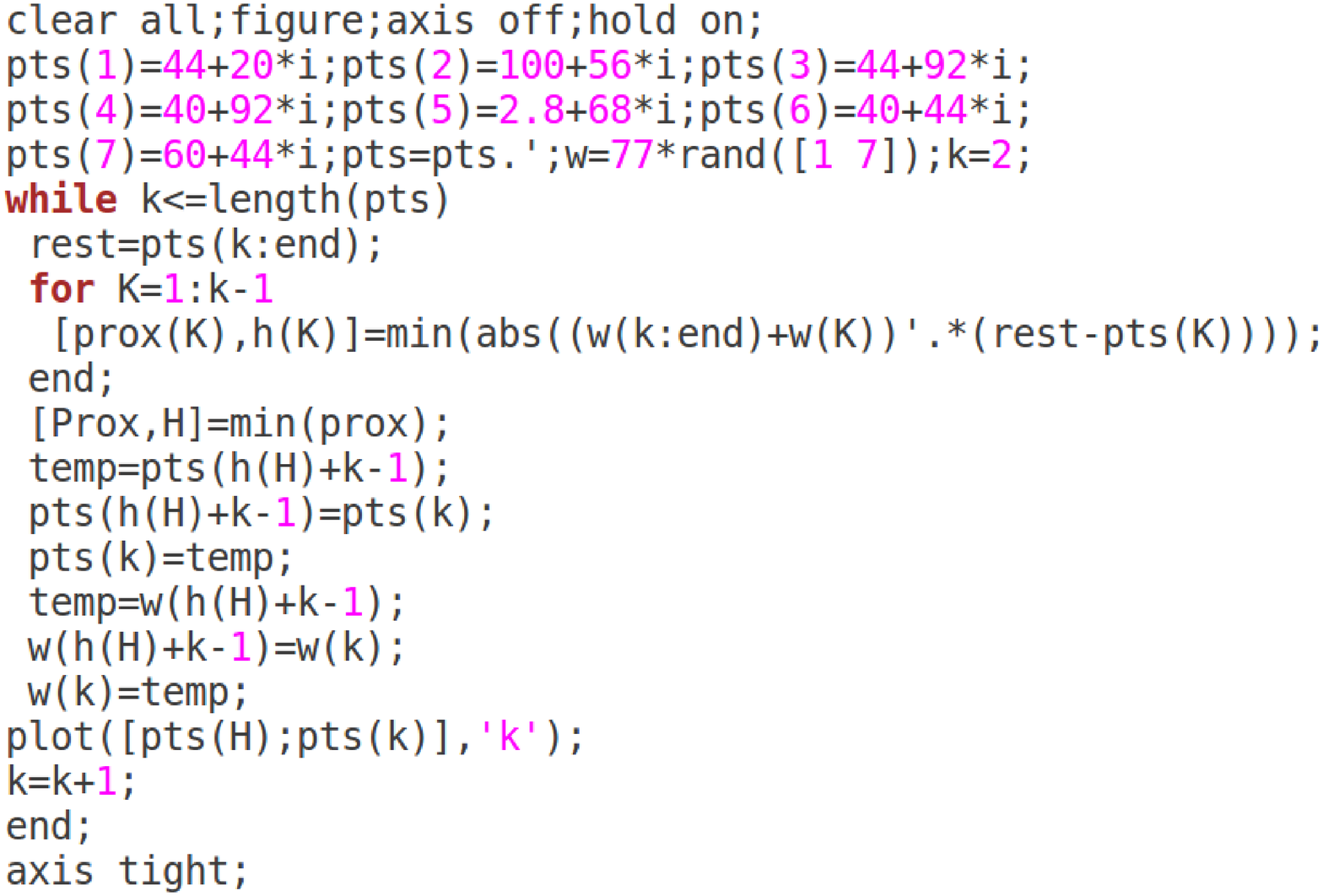}

\caption{Algorithm for test on the WMSTs of the seven vertices in Fig.~\ref{nev}.}
\label{fift}
\end{figure}

Of course, we shall be grateful to whoever forwards us some clue about either Assumption~\ref{asspt2} or~\ref{asspt1}, which was also submitted to tests but here we shall omit them for the sake of concision.

Now we draw some comparisons with GeoSteiner~\cite{Geo}, which dates back from 1997 and had its most recent version 5.1 launched in 2017. To the best of our knowledge GeoSteiner is the only exact and open-source algorithm to find non-weighted Steiner minimal trees. At \url{http://www.geosteiner.com} the reader can download its C-code, instructions and manual. However, it heavily uses Linear Programming methods and they are known to be fast only for random distributions. For instance, poor convergence of the Simplex Method is given by the classical Klee-Minty cube~\cite{KM,GB}. That is why GeoSteiner will converge very slowly if $\sharp V\cong40$ with elements that form a pattern. 

For tests our platform is 7GB of RAM, 960G of HD, microprocessor Intel Core i5 2.5GHz, and operating system Linux Ubuntu 16.04. With this setting GeoSteiner takes 68.84s to generate Fig.~\ref{comps}(a). By choosing adequate weights for $\P$ as shown in Fig.~\ref{comps}(b) we get $\T$, the heuristic WSMT in 0.01367s.

\begin{figure}[!htb]
\centering
\begin{minipage}{.5\textwidth}
\flushleft
\includegraphics[scale=0.45]{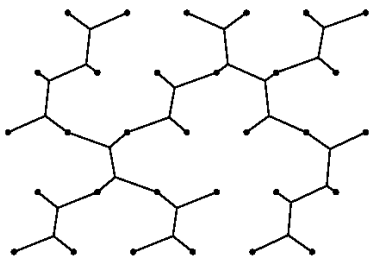}
\end{minipage}%
\begin{minipage}{.5\textwidth}
\flushright
\includegraphics[scale=0.35]{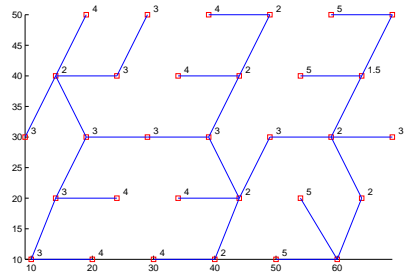}
\end{minipage}
(a)\hfil\hfil(b)

\bigskip

\includegraphics[scale=.33]{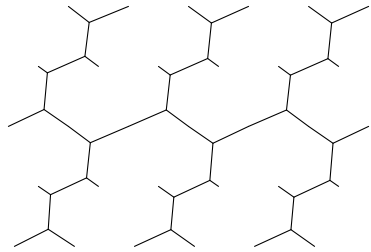}\\

(c)

\caption{Results of (a) SMT by GeoSteiner; (b) WMST by our preprocessing; (c) WSMT by our heuristic.}
\label{comps}
\end{figure}

Though Fig.~\ref{comps}(c) shows a WSMT our program also outputs the Euclidean total length $|\T|=277.2$, which is the minimum as printed by GeoSteiner. By the way, our program also gives $|\P|=318.1$, $||\P||=841.1$ and $||\T||=626.2$, both with ratios above $\sqrt{3}/2$. 

\section{Conclusion}
\label{conc}

It is known that Steiner trees can also be treated in 3d, for instance to minimise costs of underground mining~\cite{AB}, and also in the periodic approach to model crystalline and molecular connections~\cite{AK}. But 2d Steiner trees cannot be doubly periodic, only quasi-periodic instead. For instance, Fig.~\ref{almost}(a) shows a generator taken from \url{https://people.eng.unimelb.edu.au/brazil}, and it consists of copies of two fundamental sub-trees: with yellow and red nodes, represented as black and grey squares in Fig.~\ref{almost}(b), respectively. From the diagram of Fig.~\ref{almost}(b) one easily understands how to continue the quasi-periodic ST: take four copies and connect then with a black square, then repeat the process indefinitely. 

\begin{figure}[!htb]
\centering
\begin{minipage}{.5\textwidth}
\flushleft
\includegraphics[scale=0.65]{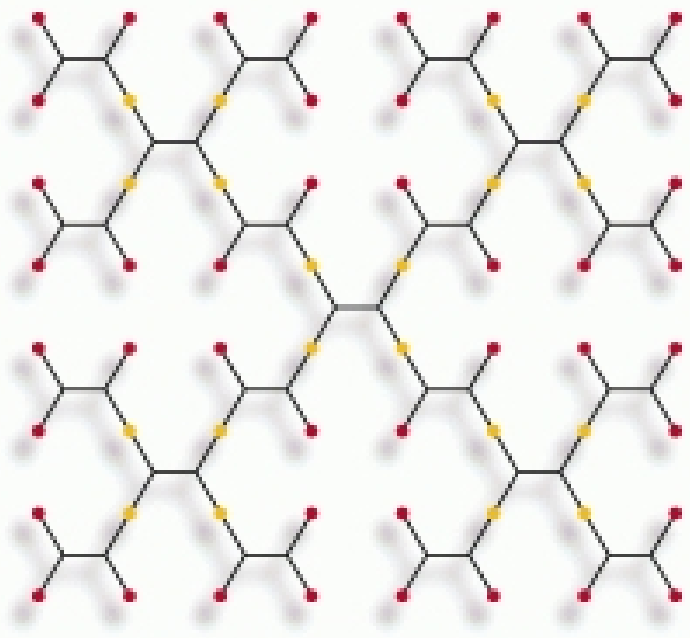}
\end{minipage}%
\begin{minipage}{.5\textwidth}
\flushright
\includegraphics[scale=0.37]{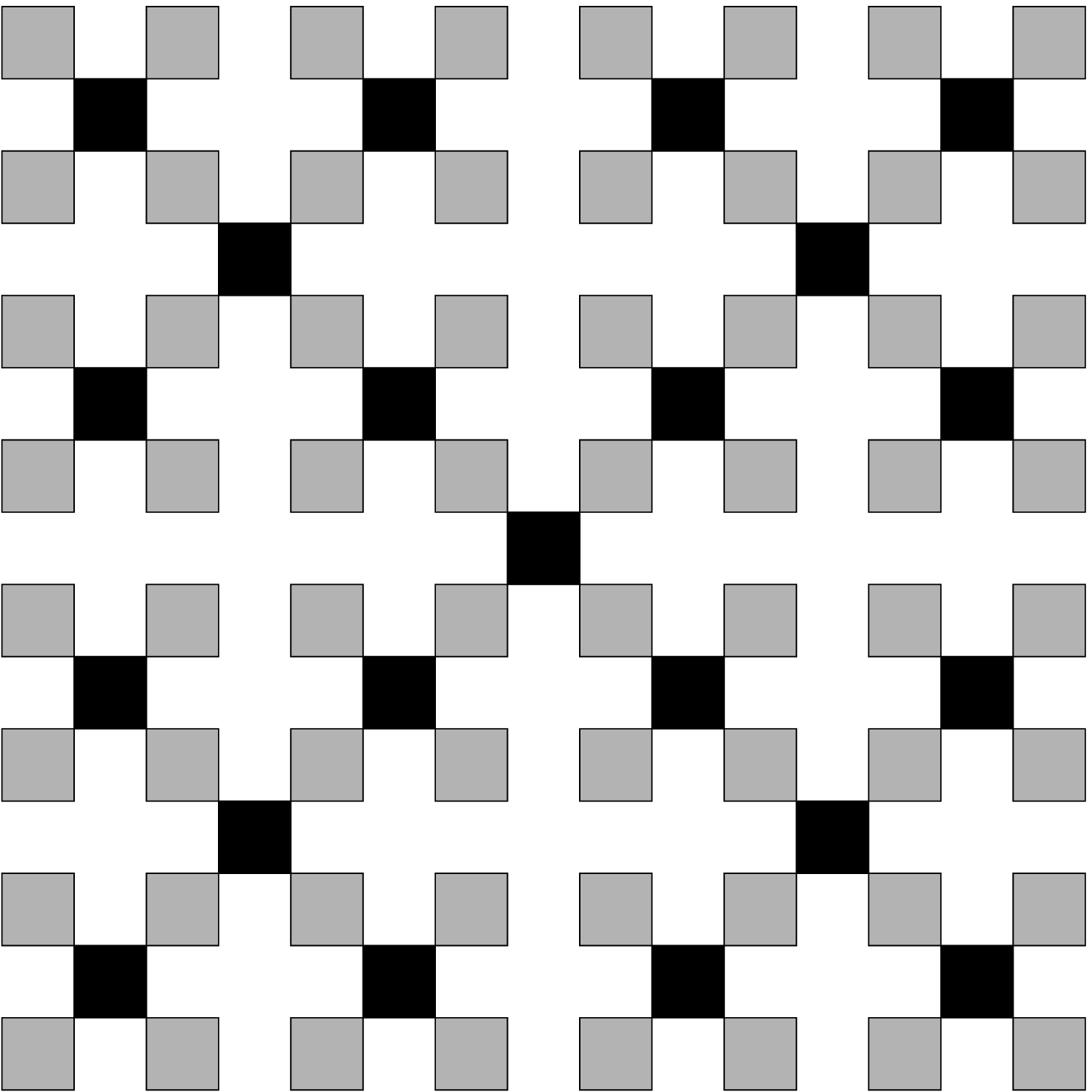}
\end{minipage}

\bigskip

(a)\hfil\hfil(b)\\

\caption{(a) a detail of a quasi-periodic ST; (b) explanatory diagram.}
\label{almost}
\end{figure}

We can also consider STs in non-Euclidean 2d and 3d spaces. For instance, in~\cite{CNR} we showed a practical application of WSMTs as in Definition~\ref{def2} for cities of a flat land whose soil is free from barriers like groundwater and rocky earth. Besides their simple physical interconnection weights were taken as extra costs like land taxes and local maintenance expenses. However, a more realistic approach should consider relief and barriers, namely the cities taken as points in the graph of a function $f:D\to\R$, $D\subset\R^2$. Any shortest connection between two points becomes a geodesic segment on graph($f$), which in many cases can fail to be unique.

Even in Euclidean spaces we can take a non-standard metric, as for the rectilinear 2d and 3d STs, which are useful in VLSI-design with millions of nodes. But already just hundreds of nodes represent an important approach (e.g. in sound and video card design)~\cite{FK,Z}.

For all these variations the greatest advantage of Evolver is its portability. In this work we have presented our algorithms already implemented in MATLAB/Octave (preprocessing) and Evolver (main program). Of course, we aimed at both reproducing the classical soap film experiment and finding the WSMT. But another important contribution is the use of Evolver with its built-in minimisation routines, which enabled us to obtain a very short code. Future works will start from this source code and take advantage of Evolver's facilities in adapting it to many other contexts.

\section*{Acknowledgement}

The first author has been supported by CAPES through his master's scholarship No. 88882.451694/2019-01.

\bibliography{smpt}

\end{document}